\documentclass[a4paper,UKenglish,cleveref, autoref, thm-restate]{lipics-v2021}

\pdfoutput=1 
 \hideLIPIcs  

\title{Selfish Mining under General Stochastic Rewards} 

\author{Maryam Bahrani}{Ritual}{maryam.bahrani.14@gmail.com}{https://orcid.org/0009-0005-6040-7755}{}
\author{Michael Neuder}{Ethereum Foundation}{michael.neuder@ethereum.org}{https://orcid.org/0000-0001-8813-6524}{}
\author{S. Matthew Weinberg\footnote{During Professor Weinberg's development of this paper, he participated as an expert witness on behalf of the State of Texas in ongoing litigation against Google (the ``Google Litigation'').}}{Princeton University}{smweinberg@princeton.edu}{https://orcid.org/0000-0001-7744-795X}{Supported by NSF CAREER Award CCF-1942497.}

\authorrunning{M. Bahrani, M. Neuder, M. Weinberg} 

\Copyright{Maryam Bahrani, Michael Neuder, S. Matthew Weinberg} 

\ccsdesc[500]{Applied computing~Electronic commerce}
\ccsdesc[300]{Security and privacy~Distributed systems security}

\keywords{Proof-of-Work, Selfish Mining, MEV}

\category{} 

\nolinenumbers 

\EventEditors{Zeta Avarikioti and Nicolas Christin}
\EventNoEds{2}
\EventLongTitle{7th Conference on Advances in Financial Technologies
	(AFT 2025)}
\EventShortTitle{AFT 2025}
\EventAcronym{AFT}
\EventYear{2025}
\EventDate{October 8--10, 2025}
\EventLocation{Pittsburgh, PA, USA}
\EventLogo{}
\SeriesVolume{354}
\ArticleNo{28}

\usepackage{enumitem}
\usepackage{algorithmic}
\usepackage{dsfont}
\usepackage[ruled]{algorithm2e} 

\newcommand{\nxt}[1]{\texttt{Next}({#1})}
\newcommand{\nxtb}[1]{\texttt{Next\_Broadcaster}({#1})}
\newcommand{\Rc}{\texttt{Claimed}}
\newcommand{\chain}{\texttt{Chain}}
\newcommand{\ts}{\texttt{Timestamp}}
\newcommand{\tstamp}{\texttt{Timestamp}}

\newcommand{\statezeronosp}{\texttt{State 0}}
\newcommand{\stateone}{\texttt{State 1}}
\newcommand{\stateonenosp}{\texttt{State 1}}

\newcommand{\stateinosp}{\texttt{State i}}
\newcommand{\Ropt}{R_{\text{opt}}}
\newcommand{\Bopt}{B_{\text{opt}}}
\renewcommand{\vec}[1]{\overset{\smash{\raisebox{-0.3ex}{\tiny$\rightharpoonup$}}}{#1}}

\begin{document}
	
	\maketitle
	\begin{abstract}
		
		Selfish miners selectively withhold blocks to earn disproportionately high revenue. The vast majority of the selfish mining literature focuses exclusively on block rewards. 
		\cite{carlsten2016instability} is a notable exception, observing that similar strategic behavior is profitable in a zero-block-reward regime (the endgame for Bitcoin's quadrennial halving schedule) if miners are compensated with transaction fees alone.
		Neither model fully captures miner incentives today. The block reward remains $3.125$ BTC, yet some blocks yield significantly higher revenue. 
		For example, congestion during the launch of the Babylon protocol in August 2024 caused transaction fees to spike from 0.14 BTC to 9.52 BTC, a $68\times$ increase in fees within two blocks.
		
		Our results are both practical and theoretical. Of practical interest, we study selfish mining profitability under a combined reward function that more accurately models miner incentives. This analysis enables us to make quantitative claims about protocol risk (e.g., the mining power at which a selfish strategy becomes profitable is reduced by $22\%$ when optimizing over the combined reward function versus block rewards alone) and qualitative observations (e.g., a miner considering both block rewards and transaction fees will mine more or less aggressively respectively than if they cared about either alone). These practical results follow from our novel model and methodology, which constitute our theoretical contributions. We model general, time-accruing stochastic rewards in the Nakamoto Consensus Game, which requires explicit treatment of difficult adjustment and randomness; we characterize reward function structure through a set of properties (e.g., that rewards accrue only as a function of time since the parent block). We present a new methodology to analytically calculate expected selfish miner rewards under a broad class of stochastic reward functions and validate our method numerically by comparing it with the existing literature and simulating the combined reward sources directly.
	\end{abstract}
	
	\newpage

\section{Introduction}
Blockchain consensus mechanisms rely on incentives to coordinate behavior. To remain safe and live, crypto-economic systems require a majority (as in Proof-of-Work) or a super-majority (as in Proof-of-Stake) of participants to adopt the protocol-specified (sometimes referred to as ``honest'') actions. Selfish mining \cite{eyal2013majority} first demonstrated that this honest behavior might not be incentive compatible for the rational miner who could earn a disproportionately large fraction of block rewards by selectively delaying the publication of their blocks. In the ensuing decade, a rich literature around strategic behavior in consensus protocols developed (e.g., in Ethereum Proof-of-Stake \cite{neuder2021low, schwarz2022three,neu2022two}). The vast majority of this literature focuses on strategies that optimize for the portion of the protocol-assigned rewards earned by the agent. These rewards, sometimes referred to as ``protocol issuance'' or ``consensus rewards,'' have historically accounted for nearly all of the value in consensus participation; this is no longer true.

As modern blockchains gain usage and facilitate more significant economic activity, their decentralized applications generate revenue. Consensus participants can collect some of this revenue through the block producer's ability to arbitrarily re-order, insert, and delete transactions when they are elected leader; \cite{daian2019flash} introduces this concept as Miner/Maximal Extractable Value (abbr. MEV). MEV has been studied theoretically and measured empirically, leading to significant changes in blockchain design. Ethereum best exemplifies this, as over 90\% of its blocks are built using a public, open-outcry block-building auction. The motivation for this auction is grounded in the notion of ``fairness'' of validator rewards. By creating a transparent market for buying and selling transaction orderings, each consensus participant should earn about the same amount of MEV – a principle originally encoded into consensus rewards, which are proportional to investment (measured in either work or stake). 

A separate line of literature studies strategic behavior in decentralized finance (abbr. DeFi), which represents another source of rewards generated at the application layer. For example, loss-versus-rebalancing \cite{lvr} (abbr. LVR) measures the amount of loss incurred by liquidity providers in decentralized exchanges as arbitrageurs balance the price of the decentralized exchange against an infinitely deep centralized exchange. These losses are precisely the profit available to those performing the arbitrage. This model completely abstracts the block creation and consensus processes, only considering the profits available to traders. In reality, the block producer has the final say over the transactions in their block, resulting in a large portion of this value flowing back to the consensus participants themselves.

The perspectives of the selfish mining, MEV, and DeFi literatures are incomplete in isolation. The co-mingling of revenue across the consensus and application layers necessitates a more precise model of rewards and their impact on strategic behavior, as demonstrated in the following real-world examples.

\begin{example}[The launch of Bablyon]\label{ex:babylon}
	On August 22, 2024, the Babylon \cite{tas2023bitcoin} protocol launched on Bitcoin. The launch allowed \texttt{BTC} tokens to be ``locked'' through a transaction processed on the chain. With a cap of 1000 \texttt{BTC}, demand for transaction inclusion spiked as people rushed to be among the first to lock their tokens. This congestion led to a $68\times$ increase in transaction fee revenue from $0.138$ to $9.515$ \texttt{BTC} between parent and child blocks \texttt{857909}, \texttt{857910}; over the four block range of \texttt{857908} to \texttt{857911}, the fee revenue increased by $500\times$ from $0.031$ to $15.551$ \texttt{BTC} \cite{bitcoinblock2025}. This immense growth in transaction fees persisted for only seven blocks, with an average per-block fee revenue of 9.64 \texttt{BTC}, after which the protocol reached its cap and fees returned to baseline levels. For those seven blocks, the block reward of $3.125$ \texttt{BTC}, which normally represents nearly the entire source of miner revenue, was only 25\% of the rewards claimed. Despite the limited scope of Bitcoin applications, Babylon exemplifies how non-protocol-specified rewards can dramatically distort miner incentives.    
\end{example}

\begin{example}[The ``Low-Carb Crusader'']\label{ex:lowcarb}
	Proof-of-Stake differs from Proof-of-Work in that it requires stakers to explicitly lock up capital to participate in the system. While Proof-of-Work is limited only to incentivizing miners with positive rewards, Proof-of-Stake enforces a subset of the protocol rules through the credible threat of destroying the capital owned by a misbehaving staker. Historically, this stick has served as an effective deterrent, but on April 2, 2023, an attacker referred to as the ``Low-Carb Crusader'' exploited a piece of infrastructure in the Ethereum protocol motivated by application layer-generated rewards. By tricking a server facilitating the block building auction referenced above, the attacker accessed private transaction data, which they exploited to $20$ million \texttt{USD} of MEV \cite{lowcarb}. In the Ethereum specification, this behavior violated the rules and thus was subject to a slashing penalty of $1$ \texttt{ETH} ($2600$ USD at current prices) levied against the attacker's stake. Clearly, the consensus reward and penalty mechanism could not account for this magnitude of profit arising from the application layer. This example demonstrates the risk facing consensus mechanisms, where non-honest behaviors are incentivized with multi-million dollar exogenous rewards originating from the application layer. 
\end{example}

These examples shows how the economic value generated in the application layer bleeds into the consensus layer rewards; see \Cref{app:timing-games} for a discussion on ``timing games,'' which is another source of revenue for consensus participants (particularly in Proof-of-Stake). To fully understand consensus incentives, a more general model for rewards is needed. In particular, a more accurate view of rewards would capture the aggregate incentives for following a specific strategy under many distinct revenue streams. The present work was motivated by that reality and takes the first step toward modeling general stochastic rewards in longest-chain protocols.

\subsection{Related work}\label{subsec:relwork}
Combining the proportion of block rewards and the linear-in-time transaction fee models of \cite{eyal2013majority} and \cite{carlsten2016instability} was the initial motivation for this work. We build upon their Markov Chains to analyze expected attacker rewards and study the $\beta$-cutoff strategies for selfish mining. As previously noted, neither work captures the state of the world in 2025; the fundamental question of `how vulnerable is Bitcoin to Selfish Mining now?' remains unanswered and of interest to the research community; this work seeks to address this gap in the literature.
\cite{ZurAET23} demonstrates how large ``whale transaction'' fees in conjunction with the standard block rewards may result in attacker profitability at lower hashrates. They use reinforcement learning to approximate the optimal policy and profit for attackers.
We also model these rewards as granting bonus value to blocks depending on the outcome of a Bernoulli trial. Our framework (\Cref{sec:generalstatic}) accommodates much more general rewards, and our instantiation (\Cref{sec:multiplerewards}) includes a third source – linear-in-time transaction fees. Further, we analytically solve for the profitability of strategies rather than approximating them.

The selfish mining literature has grown extensively in the past dozen years; see \cite{finkbeiner2025sok} for a recent survey. \cite{nayak2016stubborn,sapirshtein2017optimal,hou2019squirrl} generalized the basic selfish mining strategy to broader strategy spaces. \cite{tsabary2018gap} studied the effect of the relative sizes of block rewards and transaction fees and the impact on miners' decisions on when to mine; \cite{goren2019mind} extended that analysis and show that if all miners are rationale, the equilibrium hash rates will be far below the maximal capacity. \cite{brown2019formal} demonstrated that longest chain Proof-of-Stake protocols would also be vulnerable to selfish mining – a result instantiated through numerous selfish strategies in various staking protocols: \cite{neuder2021low,schwarz2022three,neu2022two} in Ethereum, \cite{ferreira2022optimal,ferreira2024computing} in Algorand's cryptographic self-selection, \cite{neuder2019selfish,neuder2020defending} in Tezos. We extend our model of the Nakamoto Consensus Game from \cite{bahrani2024undetectable}, which studies the detectability of selfish mining in Proof-of-Work.

MEV is one of the most relevant topics existing blockchains are reckoning with; we focus how MEV impacts consensus mechanisms. \cite{daian2019flash} coined the term and introduced many of the key properties of MEV in permissionless systems. \cite{yang2022sok} systematized MEV strategies and proposed mitigations. \cite{bahrani2024transaction,capponi2024proposer,gupta2023centralizing} focused on the centralizing nature of MEV and how Ethereum's block building market is implemented through ``Proposer-Builder Separation.'' \cite{oz2023time, schwarz2023time} studied timing games and their impact on consensus. \cite{yang2024decentralization,oz2024wins} empirically analyzed Ethereum block builders and how the market structure has evolved. We also draw on the DeFi literature when considering application-generated revenue for consensus participants. We focus on arbitrage profits as captured in LVR \cite{lvr}. \cite{lvr-fees} extends the original model to capture trading fees. 

\subsection{Summary of results}
We partition our results into two sets: practical and theoretical.

Despite the analysis of selfish mining under block rewards and transaction fees alone being several years old (\cite{eyal2013majority,carlsten2016instability}), there remains a glaring hole in the literature to study selfish behavior under the combined rewards. Quoting from \cite{finkbeiner2025sok}, a recent selfish mining SoK, ``we find that only 3 works include transaction fees in their modeling; 2 consider both block rewards and transaction fees.'' 
As described in \Cref{subsec:relwork}, \cite{ZurAET23} model transaction fees as ``whale transactions'' instead of linear-in-time; 
\cite{grunspan2018profitability} approximate transaction fees using the average amount of time in each block; thus, they simply increase the size of the fixed block rewards. 
\textbf{Our first contribution is an analysis of selfish mining under the combined model of block and transaction fee rewards; this practical contribution helps paint a more accurate picture of the selfish mining in Bitcoin under a realistic aggregate reward function.} 

\Cref{subsec:numerical} contains numerical results and discussion (see \Cref{fig:interpolation}) for the basic combination of fees and block rewards, along with other aggregate reward functions. Critically, as demonstrated in \Cref{fig:threshold-alphas}, the protocol risk depends greatly on the reward model. For example, we show that the threshold at which an attack becomes profitable decreases by 22\% when considering the two rewards together instead of only block rewards. 
Additionally, our plots allow us to make qualitative observations about miner behavior under different reward schemes. For example, a miner considering both block rewards and transaction fees will mine more or less aggressively, respectively, than if they cared about either alone. We confirm these analytical results through simulations (\Cref{fig:sims}) and by directly comparing them to existing literature (\Cref{app:block-rews-only,app:lin-rews-only}).

To derive the aforementioned practical results, we develop a set of theoretical results that may be of independent interest.  
\textbf{We present (i) a model of the Nakamoto Consensus Game with general stochastic reward sources, (ii) a novel methodology to analytically solve for a selfish miner's rewards, and (iii) a natural set of reward function properties.}
\Cref{sec:prelims} describes the new structure we impose on the NCG and how general, time-accruing reward sources interact with difficulty adjustment, which we must explicitly account for. 
Further, unlike previous work,\footnote{With the sole exception of \cite{ZurAET23}, which studies a narrow set of random rewards – see \Cref{subsec:relwork} for discussion} our reward functions can be stochastic. Namely, we study a much more general class of \textit{static} rewards (\Cref{def:static}), which we define as functions that accrue randomly and independently only as a function of time since the parent block. 
Calculating attacker profits under these general reward sources requires the novel, path-counting technique presented in \Cref{sec:generalstatic}. The practical results (\Cref{sec:multiplerewards}) described in the previous paragraph follow as a corollary since the combination of block and transaction fee rewards is static. 

Lastly, we characterize a natural set of reward function properties motivated by existing blockchains (\Cref{subsec:properties} and \Cref{app:max-blocks-persistent}). We illustrate these properties through two extensive case studies. \Cref{sec:examples} examines transaction fees and describes how different assumptions about block size, transaction patience, and arrival rate manifest in very different reward functions captured by our properties. \Cref{app:lvr-examples} focuses on arbitrager profits under various assumptions about price trajectories and leader-election mechanisms.

\section{Preliminaries and model}\label{sec:prelims}

We start by defining a stylized model of Proof-of-Work mining with general stochastic rewards. This necessitates some crucial differences between our model and previous selfish mining literature. For example, general rewards can be sensitive to specific inter-block times, requiring explicit modeling of difficulty adjustment. \Cref{sec:notes-on-model} discusses these differences in detail.

\subsection{Nakamoto Consensus Game with general rewards}\label{subsec:NCG}

Let $M$ denote the set of $n$ miners, where miner $m\in M$ has hashrate $\alpha_m$. 

\paragraph*{Views.}
At any time $t$, there is a public \emph{view} $V_t$, consisting of the ``state'' of the blockchain known to all miners at time $t$. This view includes all blocks that have already been broadcast, their creation times, and the identity\footnote{Real-world blockchains are often pseudonymous, and the ``identities'' of miners refer to their public keys.} of their creators in $M$. It also includes the content of each block, which contains enough information to compute the values of all variables and account balances in every block across forks. For each block $B$ in a view, we have $\tstamp(B)$, the time\footnote{Timestamp here refers to the actual creation time of the block, rather than a reported time stated by the miner.} that the block was produced.

At any time $t$, there is also a private view $V_t^m$ for each miner $m$ that includes $V_t$ and potentially some additional blocks $m$ knows about that are unknown to all other miners (e.g., a private fork). We assume that miners don't selectively exclude a subset of miners when they broadcast, and all broadcasting happens instantaneously (e.g., no eclipse attacks \cite{heilman2015eclipse}). As a result, $V_t^m$ will only include $V_t$ and any blocks mined by $m$ that have not yet been broadcast (along with their contents).

\paragraph*{General Rewards.} Miners are rewarded for creating blocks on the eventual longest chain in the form of block rewards (a fixed value issued once per block), fees from included transactions, and potentially additional revenue stemming from their monopolistic control over the content of the block (MEV). The size of this reward can be different across blocks and might be stochastic. We abstractly model these rewards as a function $R$.

Fix a time $t$, a view $V$, a block $B$ in $V$, and a miner $m$. We use $r$ to capture any exogenous randomness that could impact the value of blocks that a miner creates (e.g., the launch of a protocol that could create large amounts of congestion and resultingly higher transaction fees as in \Cref{ex:babylon}). We denote by $\mathcal{B}^m(t,V,B,r)$ the set of \emph{valid} blocks that $m$ can create.
Because not all views are achievable under a specific realization of the randomness $r$, when we invoke a view $V$ together with $r$, we implicitly restrict $r$ such that $V$ is realizable.

\begin{definition}[Reward Function]
	A \emph{reward function} $R^m$ for miner $m$ takes as input a time $t$, a view $V$, a block $B$ in $V$, randomness $r$, as well as a block $B'\in\mathcal{B}^m(t,V,B,r)$, and outputs a real number,
	\begin{align*}
		R^m(t,V,B,r,B') \to \mathbb{R}.
	\end{align*}
\end{definition}

The output of $R^m$ can be interpreted as the amount of reward collected by $m$ for creating a block $B'$ that extends $B$ in $V$ at time $t$ given randomness $r$, assuming $B'$ ends up on the eventual longest chain.

We allow different miners to have different reward functions to keep the model general. This per-miner reward can capture miner heterogeneity (e.g., from private order flow or better trading strategies). For the properties we define in \Cref{subsec:properties} and the selfish mining analysis in \Cref{sec:generalstatic,sec:multiplerewards}, however, we restrict our study to \emph{miner-independent} (see \Cref{def:minerindependent} below) reward functions.

\paragraph*{Miner Strategies.} Each miner $m$ has a strategy that takes as input a time $t$, a view $V_t^m$, and the reward $R^m(t,V_t^m,B,r,B')$ for extending each block $B\in V_t^m$ by a valid block $B'\in\mathcal{B}^m(t,V_t^m,B,r)$, and outputs
\begin{itemize}
	\item a block $B\in V_{t}^m$ to mine on,
	\item contents of the next block $B'\in\mathcal{B}^m(t,V_t^m,B,r)$, and
	\item a (potentially empty) subset of blocks in $V^m_t\setminus V_t$ to broadcast.
\end{itemize}
For each miner $m$, we denote by $\nxt{m,t,V_t^m,r}$ the first time after (or equal to) $t$ that $m$ broadcasts a block assuming their private view remains $V_t^m$, and by 
\begin{align*}
	\nxtb{t,r}:=\underset{m\in M}{\arg\min}\{ \nxt{m,t,V_t^m,r} \},
\end{align*}
the identity of the next miner to broadcast after (or at $t$), breaking ties arbitrarily. We use these functions to determine the ordering of broadcasters as the game progresses (see details in Appendix – \Cref{alg:ncg-updates}).

Note that miner strategies cannot directly observe the randomness $r$ but might indirectly depend on it through the realizations of $R^m$ and $\mathcal{B}^m(t,V_t^m,B,r)$, all of which take as input the same randomness $r$. While we focus on deterministic miner strategies in this paper, our model can easily be extended to account for randomized behavior.

\paragraph*{Nakamoto Consensus Game (NCG)}

The Nakamoto Consensus Game describes how views evolve given a fixed set of miner strategies. We model the game after difficulty has already been adjusted according to these strategies, resulting in a stable orphan rate $\lambda$,\footnote{See \Cref{sec:notes-on-model} for extended discussion and \cite{bitcoin-book} for a more comprehensive overview of difficulty adjustment is used in the Bitcoin protocol.} and we normalize time so that the average block time is 1. We let time $0$ refer to a point after which the difficulty of mining puzzles remains constant. We further assume that miners only extend blocks created after time $0$.

Prior to the game, we draw the following random variables independently:\footnote{See \Cref{app:further-model} for a discussion of why we can assume independence.}
\begin{itemize}
	\item \textit{Miner selection} – A sequence of miners $\vec{m}\in M^\mathbb{N}$, where $m_i$ is the creator of the $i^{th}$ block. For each $i$, $m_i$ is selected independently such that it equals $m\in M$ with probability $\alpha_m/\sum_{j=1}^n\alpha_j$.
	\item \textit{Block times} – A sequence of block creation times $\vec{t}\in\mathbb{R}^\mathbb{N}$, where $t_0:=0$, and the duration $t_{j}-t_{j-1}$ for $j\geq 1$ is drawn i.i.d. from an exponential distribution with rate $1/(1-\lambda)$.
	\item \textit{Remaining randomness} – The randomness $r$.
\end{itemize}
Initially, there is some public view $V_0$ but no hidden blocks, so $V^m_0=V_0$ for all $m\in M$, where $V_0:=\{B_0\}$ is the view containing a single genesis block $B_0$ such that $\tstamp(B_0)=0$.
Starting with $j=1$ (the variable used to index the miners $\vec{m}$ and block times $\vec{t}$) and $t=0$, we check if there are new blocks to broadcast before updating the block that each miner is building on based on the contents of the pre-determined strategy. See \Cref{app:algo} for the  procedure to carry out the NCG and for a note on ensuring uniqueness of the longest chain.

\subsection{Notes on model}\label{sec:notes-on-model}

We briefly summarize how we model difficulty adjustment below. 
See \Cref{app:further-model} for an extended comparison of our model to previous work and the role of independence in the randomness of rewards.

\paragraph*{Difficulty adjustment.}
In practice, mining involves solving computational puzzles with adjustable difficulty. Since miners can enter (or exit) permissionlessly, the total hashrate of all miners can vary over time, resulting in varying block production rates. The protocol varies the difficulty of these puzzles based on timestamps of recent blocks, targeting a fixed average inter-block time. In Bitcoin, the difficulty updates once every difficulty \emph{epoch} (2016 blocks/roughly every two weeks assuming ten-minute block times) by the \emph{difficulty adjustment algorithm (DAA)}. The difficulty of extending any blocks is the same within an epoch, except for forks across the epoch boundary. Note also that forks are rarely longer than a few blocks, so this represents an insignificant fraction of the blocks in an epoch.

Fixing a set of miner strategies, one can compute the expected fraction of blocks per epoch that do not end up on the longest chain. We assume the difficulty adjusts based on this expected value (rather than directly modeling per-epoch updates described above) and calculate the profitability of various strategies under this new difficulty. Specifically, we calculate the expected orphan rate $\lambda$ (\Cref{lem:lambda}), which implies the difficulty-adjusted rate of block production is $1/(1-\lambda)$. This corresponds to blocks on the longest chain growing at an average rate of 1. 

\section{Reward functions: properties and examples}\label{subsec:properties}

Recall that miner strategies take as input the amount of reward available for extending each existing block at time $t$, as specified by the reward function $R$, and make decisions about where to mine, what to include, and what to broadcast accordingly. This section defines a set of natural properties that reward functions might have. 
In \Cref{sec:examples}, we motivate these properties with an extensive case study on transaction fees, one of the primary revenue sources observed empirically to date. See \Cref{app:lvr-emphermeral,app:max-blocks-persistent,app:persistent-reward-examples,app:lvr-examples} for additional properties and a second case study on LVR, a prominent source of revenue on chains with significant DEX volume.

While we define these properties in the context of the NCG in this paper, we believe their applicability extends far beyond Proof-of-Work and selfish mining. Our framework can be used to characterize rewards and their implications for the incentives of consensus participants across blockchain protocols.

Recall that in the NCG, given a set of miner strategies, three independent random variables $\vec{t},\vec{m},r$ are drawn and used to compute a set of views $V_t^m$ for all miners $m$ and all times $t$.
Let $\mathcal{V}_t^m$ be the support $V_t^m$, meaning the set of views achievable at time $t$ for \emph{some} realization of $\vec{t},\vec{m},r$. Initially, $\mathcal{V}_0^m=\{V_0\}$ for all $m$, where $V_0:=\{B_0\}$ is the view containing a single genesis block $B_0$ such that $\tstamp(B_0)=0$. Miner strategies in the NCG take the realization of a reward function as input. That is, at time $t$, miner $m$ sees the reward $R^m(t,V_t^m,B,r,B')$ for extending each block $B\in V_t^m$ by a valid block $B'\in\mathcal{B}^m(t,V_t^m,B,r)$.

A miner-independent reward function yields the same value for the block regardless of who created it. This 
corresponds to a setting where all miners have access to the same set of rewards (e.g., the common value setting), and thus, we drop the superscript $m$. In practice, some reward sources may be heterogeneous between block producers (e.g., from private order flow or from differing abilities to extract MEV \cite{bahrani2024centralization}). All reward functions considered in this paper will be miner-independent, but the properties can be readily generalized by tracking the subset of miners with access to each reward source. See \Cref{sec:conclusion} for a discussion of extending this work.

\begin{definition}[Miner-Independent Rewards] \label{def:minerindependent}
	A reward function $R$ is \emph{miner-independent} if for all times $t$, all miners have the same set of valid views, the same set of valid blocks extending each block in those views, and equal rewards from any such valid block.\footnote{\label{foot:bijection}Technically, since blocks include information about their creator, it would be more accurate to say that there is a bijection between the set of valid views/blocks for any pair of miners. We overlook this formality to simplify notation.} Formally, $R$ is miner-independent if for all $t$, and all $m,m'\in M$, 
	\begin{itemize}
		\item $\mathcal{V}_t^m=\mathcal{V}_t^{m'}$,
		\item for all $V\in\mathcal{V}^m_t$, all blocks $B$ in $V$, and all $r$, we have $\mathcal{B}^m(t,V,B,r)= \mathcal{B}^{m'}(t,V,B,r)$,
		\item for all $V\in \mathcal{V}^m_t$, all $r$, all parent blocks $B$ in $V$, and all valid blocks $B' \in \mathcal{B}^m(t,V,B,r)$, we have $R^m(t,V,B,r,B')= R^{m'}(t,V,B,r,B')$.
	\end{itemize}
\end{definition}

We can also characterize reward functions that grow according to the same distribution without depending on the chain's history. The following property limits the dependence of $R$ on the view. Intuitively, it says that the only relevant information in the view that affects the amount of reward in a block is the \emph{timestamp of its parent}.

\begin{definition}[View-Independent Rewards]\label{def:viewindependent}
	A reward function $R$ is \emph{view-independent} if for all times $t'<t$, any two views $V_1,V_2\in\mathcal{V}_{t'}$ such that $\tstamp(B_1)=\tstamp(B_2)=t'$ for some blocks $B_1\in V_1,B_2\in V_2$, we have:
	\begin{itemize}
		\item for all $r$, the set of valid blocks extending $B_1$ at $t$ in $V_1$ is the same as the set of valid blocks extending $B_2$ at $t$ in $V_2$, $\mathcal{B}(t,V_1,B_1,r)=\mathcal{B}(t,V_2,B_2,r)$,\footnote{Recall that when we invoke a view and randomness together as inputs to a function, we implicitly assume that the randomness could give rise to the view.}
		and
		\item for every valid block $B'\in\mathcal{B}(t,V_1,B_1,r)$, we have
		\[
		\Pr_{r,\vec{t},\vec{m}\vert V_1}[R(t,V_1,B_1, r, B')=x]=\Pr_{r,\vec{t},\vec{m}\vert V_2}[R(t,V_2,B_2, r, B')=x]
		\]
		for all $x$.
	\end{itemize}
\end{definition}

Note that fixing a view $V_1$ (resp. $V_2$) can update the distribution of the $r,\vec{t},\vec{m}$. We use the subscript $r,\vec{t},\vec{m}\vert V_i$ to refer to the posterior distribution of these random variables conditioned on $V_1,V_2$. For example, block rewards are view-independent (within the same four year halving window) because each block earns the same fixed reward from the protocol.
\Cref{rem:patient-not-vi} below is a non-example that demonstrates how transaction fees that are not fully claimed by a parent block (e.g., arising from finite block sizes) are not view-independent because the reward of the resulting child block depends on the amount of unclaimed transaction fees.  

View-independence already limits the dependence of $R$ on the view to the timestamp of the parent block. We next define a subset of view-independent rewards where the dependence on view is limited to the length of \emph{elapsed time since the parent block} (and is the same regardless of the exact parent block timestamp).

\begin{definition}[Static Rewards]\label{def:static}
	A reward function $R$ is \emph{static} if for all $\Delta>0$, all times $t_1,t_2$ and views $V_1\in\mathcal{V}_{t_1}$ and $V_2\in\mathcal{V}_{t_2}$ such that $\tstamp(B_1)=t_1-\Delta$ and $\tstamp(B_2)=t_2-\Delta$, we have:
	\begin{itemize}
		\item for all $r$, the set of valid blocks extending $B_1$ at $t_1$ in $V_{1}$ is the same as the set of valid blocks extending $B_2$ at $t_2$ in $V_{2}$, $\mathcal{B}(t_1,V_{1},B_1,r)=\mathcal{B}(t_2,V_{2},B_2,r)$, and
		\item for all valid blocks $B'\in\mathcal{B}(t_1,V_1,B_1,r)$, we have
		\[
		\Pr_{r,\vec{t},\vec{m}\vert V_1}[R(t_1,V_1,B_1, r, B')=x]=\Pr_{r,\vec{t},\vec{m}\vert V_2}[R(t_2,V_2,B_2,r,B')=x]
		\]
		for all $x$.
	\end{itemize}
\end{definition}

\Cref{rem:patientisstatic} below highlights that transaction fees are static using the \cite{carlsten2016instability} model with constant arrival rate and infinite block sizes. See appendix \Cref{rem:non-local-lvr,rem:local-lvr-static}, which demonstrate the conditions under which LVR is or is not static.

\subsection{Properties of transaction fees}\label{sec:examples}
To illustrate the value of the aforementioned properties of reward functions, we perform an extensive case study on transaction fees (see \Cref{app:lvr-examples} for a similar study but on LVR). We consider the relevant properties that arise from different assumptions about block sizes, user patience levels, and accrual rate of transactions. These examples aim to justify the properties we focus on in \Cref{subsec:properties} and motivate \Cref{sec:generalstatic,sec:multiplerewards}, which measure attacker revenue under multiple static reward sources. 

\paragraph*{Transaction fees.}
Users pay transaction fees to interact with blockchains. A mempool collects transactions as they arrive, and its state at all times is captured in our model through the realization of the randomness $r$.
Consider transactions as infinitely divisible,\footnote{\label{fn:heterogeneity}We could instead consider transactions as heterogeneous in size (e.g., as in Ethereum where transactions consume different amounts of gas) or exclusive to miners (e.g., from private order flow), but the additional complexity doesn't add anything to the qualitative observations and is thus elided.} belonging to the same mempool,$^\text{\ref{fn:heterogeneity}}$ and specifying a fee.
A valid block $B'$ mined at time $t$ and extending a parent block $B$ can include any transactions in the mempool at $t$ that are not already included in $\chain(B)$.
The corresponding reward function for a valid candidate block is the sum of the fees paid by the transactions it includes.

We call users \emph{patient} if their transactions remain valid until they are eventually included in a later block. We shorthand transactions originating from patient users as \textit{patient transactions}.

As demonstrated in the following example, we cannot claim any further structure on the patient-user transaction fee reward function without restricting the set of valid blocks. 

\begin{example}[Patient transaction fees may be view-dependent]\label{rem:patient-not-vi}
	Consider two blocks $B_1,B_2$ with the same timestamp $t'$ and with the same parent mined at $t$. $B_1$ claims all transaction fees arriving in $[t,t']$, while $B_2$ claims none. The rewards of maximizing candidate blocks $B_1',B_2'$ built on $B_1,B_2$ respectively, are different, as $B_2'$ can claim more transaction fees than $B_1'$.
\end{example}

The key observation is that miners may not claim the complete set of available transactions, thus impacting the claimable rewards of descendant blocks in that view (for more formalism, see \Cref{lem:all-claim} in the appendix). Alternatively, consider the case where each block can include all transactions (e.g., infinite block size as in \cite{carlsten2016instability}). If we additionally restrict the set of views for each miner $\mathcal{V}^m_{t'}$, we \textit{can} make the following stronger claim.

\begin{example}[Patient transaction fees are view-independent if blocks are infinite capacity and fully-claiming]\label{rem:fullyclaimingviewindep}
	Assume blocks have infinite capacity and restrict views to only include blocks that contain all available transaction fees at the time of mining.
	Then, the distribution of rewards for $B'$ built at time $t$ on parent block $B_1$ or $B_2$, which have the same timestamp $t'$, is the same. Namely, the reward is the sum of patient transaction fees arriving in the interval $[t',t]$.
\end{example}

Here, view-independence arises from the mempool fully emptying after each block is created. Thus, the reward function only depends on newly arriving transaction fees after the parent block is mined. Importantly, this reward function \textit{may not be static} (which is a stronger condition than view independence) because the transaction fee arrival rate may not be homogeneous over time. For example, some hours of the day (such as trading hours in Asia time zones) might result in higher transaction fee arrivals. Assuming a constant transaction arrival rate, we can further establish staticness.

\begin{example}[\cite{carlsten2016instability}'s model of transaction fees is static]\label{rem:patientisstatic}
	Assume 1 unit of patient transaction fees arrive per unit of time, blocks have infinite capacity, and all blocks in the view claim all available transaction fees (as in \cite{carlsten2016instability}). 
	A block $B'$ extending $B$ at time $\tstamp(B)+\Delta$ can claim any reward in $[0,\Delta]$. Therefore, this reward function is static.
\end{example}

While the previous example considers \emph{deterministic} transaction fee arrivals (1 unit of fees per unit of time), the same claim holds if the arrival rate is a random function of $r$ (but still identically distributed over time). Constant accrual, in addition to the mempool clearing, results in the reward function being independent of the timestamp of the parent block, making it static. 

Until now, we have only considered patient users. In contrast, consider \textit{impatient users}, who submit transactions that are only valid for the next block produced (e.g., by checking the height of the block they are included in before executing). We similarly shorthand these as \textit{impatient transactions}. 

\begin{example}[Identically distributed, impatient transaction fees are static]\label{rem:impatient}
	Assuming the impatient transactions arrive according to a fixed distribution over time since the parent block, this reward function is static because the mempool clears after each block.
\end{example}

Note that the mempool clearing after each block was necessary for both \Cref{rem:patientisstatic,rem:impatient} to be static. However, the clearing came about differently -- infinite block sizes in the former and impatient users in the latter. The mempool clearing is a \textit{sufficient} condition for staticness if the distribution of rewards doesn't depend on global clock time. 

Varying the assumptions on block size and user patience allows us to describe reward functions under differing models of congestion; we now consider transaction fees that are high regardless of the block size. This \textit{contentious transaction} model is motivated by the launch of Babylon (\Cref{ex:babylon}). Transaction fees may spike because there is immense demand not just for inclusion in a block but also for a specific ordering (e.g., needing to be one of the first 100 transactions of a particular type). 

\begin{example}[Bernoulli rewards are static]
	Consider contentious transaction fees modeled as independent Bernoulli trials that occur once per block height, resulting in a constant random reward of size $E$ with probability $p$. This is a static reward function.
\end{example}

In \Cref{sec:multiplerewards}, we study a variant of selfish mining under a combined reward function that includes Bernoulli rewards, linear-in-time transaction fees as in \Cref{rem:patientisstatic}, and block rewards. This combined reward function is static, which is crucial to the tractability of that analysis. See \Cref{subsec:relwork} for a discussion on the similarities between our model of Bernoulli rewards and that of \cite{ZurAET23}.
See \Cref{app:persistent-reward-examples} for examples pertaining to the properties defined in \Cref{app:max-blocks-persistent} and \Cref{app:lvr-examples} for an extended case study on LVR.

These examples showcase the properties we ascribe to general reward functions in \Cref{subsec:properties}. While these case studies allow us to demonstrate View-Independence (\Cref{def:viewindependent}) and  Staticness (\Cref{def:static}) in familiar settings, they do not cover all MEV types. As mentioned in \Cref{sec:conclusion}, we see characterizing the complete set of properties and applying them to other forms of MEV (e.g., sandwiches and liquidations) as a key direction for future work. With these properties in place, we now focus on calculating expected attacker profits from performing $\beta$-cutoff selfish mining strategies under general static reward functions. 

\section{Selfish mining with static rewards}\label{sec:generalstatic}

\Cref{sec:prelims,subsec:properties} presented our model of general stochastic rewards and created a structure around these reward functions. The subsequent sections study a specific set of miner strategies to analyze their profitability and feasibility under general static rewards (\Cref{def:static}). We examine $\beta-$cutoff selfish mining strategies \cite{carlsten2016instability}, in which the attacker determines whether or not to hide their blocks based on the amount of reward realized during the mining process.

\subsection{Mining strategies in the NCG}\label{subsec:strategies}

In the NCG defined in \Cref{sec:prelims}, miners make three decisions at each time $t$: 
\begin{enumerate}
	\item which block to extend,
	\item the contents of their next mined block, and
	\item which blocks to broadcast.
\end{enumerate}
Based on these decisions, we define the protocol-prescribed mining as honest.
\begin{definition}[Honest mining]
	The honest mining strategy is defined as,
	\begin{enumerate}
		\item mine on the longest chain,
		\item claim all available rewards, and
		\item publish every block immediately.
	\end{enumerate}
\end{definition}
In words, the honest miners always follow the longest chain and immediately share any block they find with the rest of the network. If the remainder of the network is honest, the rewards that an honest miner, $i$, controlling $\alpha_i$ fraction of the hash power is proportional to their mining power.


\cite{eyal2013majority} and \cite{carlsten2016instability} demonstrate that selfish mining is profitable for miners (even under various tie-breaking schemes) when considering \textit{only} block rewards or \textit{only} transaction fees that are linear-in-time respectively. \cite{carlsten2016instability} also introduced $\beta$-cutoff selfish mining strategies, in which the attacker mines selfishly as long as the rewards they earn on their hidden block are sufficiently small. If their rewards are larger than a threshold $\beta$, they instead broadcast immediately to avoid losing the valuable block. 
\begin{definition}[$\beta$-cutoff selfish mining \cite{carlsten2016instability}]\label{def:betacutoff}
	If there is no private chain, the attacker follows the rules:
	\begin{enumerate}
		\item mine on the public longest chain,
		\item claim all available rewards, and
		\item withhold any block found where the time since parent is less than $\beta$ (create a private chain).
	\end{enumerate}
	The third step above creates the private chain for the attacker; they transition into the following rules (same as original selfish mining):
	\begin{enumerate}
		\item mine on the private chain,
		\item claim all available rewards, and
		\item withhold any block found unless an honest block is found and the difference in length between the public chain and the private chain is $\leq 1$.
	\end{enumerate}    
\end{definition}

This strategy differs from pure selfish mining only in Step 3 under no private chain, where the attacker decides whether or not to publish based on the rewards captured in the block. Note that the strategies we consider claim all available rewards; miners could instead choose to intentionally leave  some rewards on the table to incentivize subsequent miners to build on their chain (``undercutting'' \cite{carlsten2016instability}). See \Cref{sec:conclusion} for discussion on extending our framework to a broader class of miner strategies.

Given a static reward function, we want to determine the per-unit-time expected attacker rewards from following the $\beta$-cutoff strategy as in \Cref{def:betacutoff}. We develop a new technique based on a Markov Chain similar to Figure~13 in \cite{carlsten2016instability} and Figure~1 in \cite{eyal2013majority}. 

\begin{definition}[$\beta$-cutoff Markov Chain]\label{def:markovchain}
	Consider the NCG where the $1-\alpha$ of the mining power follows the honest strategy and $\alpha$ follows the $\beta$-cutoff strategy. Then define \texttt{State i} for $i \geq 1$ where the attacker has a hidden chain $i$ blocks longer than the public chain. Let \texttt{State 0} denote the attacker having no hidden blocks and \texttt{State 0'} denote the race state between the honest and attacker forks each of length 1. Let \texttt{State 0''} denote the state immediately after the attacker publishes their private chain. 
\end{definition}

\begin{figure}
	\centering
	\includegraphics[width=\linewidth]{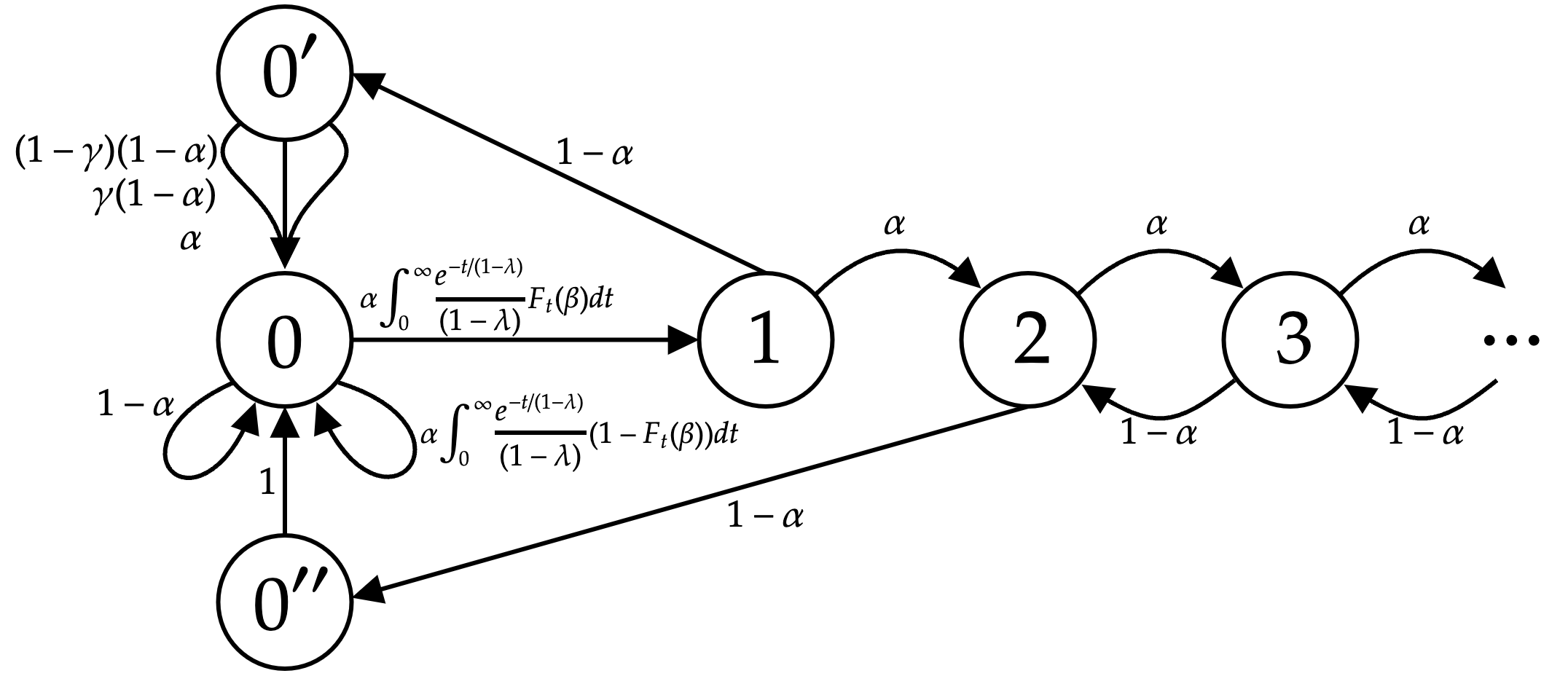}
	\caption{The Markov Chain capturing the $\beta-$cutoff strategy for miners deciding whether to publish blocks depending on the size of the static reward. $F_t(\beta)$ is the CDF of the rewards given time $t$ since the parent block, $\Pr[R(t) \leq \beta]$. The rate of the chain is $1/(1-\lambda)$, which explicitly captures the difficulty adjustment that results from a specific $\beta$-cutoff strategy.}
	\label{fig:markovgeneral}
\end{figure}

\Cref{fig:markovgeneral} depicts this Markov Chain. We now derive the transition probabilities using a general, static reward function. When considering static reward sources, notice that $R$ is only a function of the time since the parent block was mined; we hereafter denote this static reward source as $R(t)$, where $t$ is the time since the parent block. This simplification allows us to compute the probability of transitioning from \texttt{State 0} $\rightarrow$ \texttt{State 1} by comparing the expected amount of rewards earned in \texttt{State 0} conditioned on those rewards being less than $\beta$ (the cutoff threshold for publishing the block in \statezeronosp).

\begin{definition}[Static Reward CDF \& PDF]\label{def:cdf}
	For a static reward source $R$ and randomness $r$, let $F_t(x)$ denote the CDF of the reward function indexed by time $t$,
	\begin{align*}
		F_t(x) = {\Pr}_r[R(t) \leq x].
	\end{align*}
	Similarly, let $F'_t(x)$ denote the PDF of the reward function,
	\begin{align*}
		F'_t(x) = {\Pr}_r[R(t) = x].
	\end{align*}
\end{definition}
To calculate the probability of withholding the block, we integrate the probability distribution of the time until the next block multiplied by the CDF of the rewards at each time.
\begin{align}\label{eq:zerotoone}
	\Pr[\texttt{State 0} \rightarrow \texttt{State 1}] &= \alpha \int_0^\infty 
	\underbrace{\frac{e^{-t/(1-\lambda)}}{(1-\lambda)}}_{\shortstack{\scriptsize density of time}}
	\cdot \underbrace{F_t(\beta)}_{\shortstack{\scriptsize rewards $ <\beta$\\ \scriptsize by time $t$}} dt 
\end{align}
Intuitively, given a reward source $R$, this value tells us how likely it is that the rewards within an attacker block are less than $\beta$. Notice that the density function of the exponential depends on a rate parameter $1/(1-\lambda)$ (as discussed in \Cref{subsec:NCG}), where $\lambda$ the explicitly calculated orphan block rate calculated as a function of $\beta$ to account for difficulty adjustment. See \Cref{lem:lambda} for its derivation. 
Conversely, given an attacker block we can also calculate the probability that the attacker publishes the block immediately if the block rewards are be greater than $\beta$,
\begin{align}\label{eq:zerotozero}
	\Pr[\texttt{State 0} \rightarrow \texttt{State 0} &\land \text{attacker block}]  = \alpha \int_0^\infty 
	\underbrace{\frac{e^{-t/(1-\lambda)}}{(1-\lambda)}}_{\shortstack{\scriptsize density of time}} 
	\cdot \underbrace{(1-F_t(\beta))}_{\shortstack{\scriptsize rewards $ \geq \beta$\\ \scriptsize by time $t$}} dt 
\end{align}
With \Cref{eq:zerotoone,eq:zerotozero}, we construct the entire Markov chain in \Cref{fig:markovgeneral}. Note that it differs from Figure~1 in \cite{eyal2013majority} and Figure~13 in \cite{carlsten2016instability}, only in the transition probabilities from \texttt{State 0} calculated above for general static reward sources (\Cref{eq:zerotoone,eq:zerotozero}). 
As in previous work, $\gamma$ is the tie-breaking rate dictating the fraction of honest miners who mine on the attacker block after it is published, and there is a race of length-1 forks (in \texttt{State 1}). This parameter doesn't impact the $\beta$-cutoff itself and only affects the probability that the attacker fork wins the tie. See \Cref{app:stationary} for the calculation of the stationary distribution of this Markov Chain.

With the stationary distribution, we can explicitly solve for the proportion of orphan blocks, $\lambda \in [0,1]$, which in turn gives us the difficulty-adjusted rate of the Poisson process of the transitions in the Markov Chain as $1/(1-\lambda)$. This rate is \textit{faster} than the rate of canonical blocks (normalized to 1) because the orphaning process causes a reduction in difficulty.  

\begin{lemma}[Calculating $\lambda$]\label{lem:lambda}
	Let $\lambda$ measure the probability that a block produced in the Markov Chain is orphaned. Then,
	\begin{align*}
		\lambda = p_1 (1-\alpha) \left(1+\frac{\alpha}{1-2\alpha}\right).
	\end{align*}
\end{lemma}
Proof in \Cref{pr:lambda}.
With $\lambda$, the new block production rate is $1/(1-\lambda)$. This is the rate at which blocks are found by any miner (i.e., the rate of transitioning between states in the Markov Chain; \Cref{fig:markovgeneral}) assuming a constant hash rate and results in the canonical chain blocks being produced at a rate of $1$.

\subsection{Expected attacker rewards}\label{subsec:perstateattacker}
The stationary distribution alone is incomplete. To determine the attacker profit for a given cutoff strategy, we calculate their expected profit from each state and multiply those values by the stationary distribution of the Markov Chain to determine the expected profit per unit of time.
\begin{definition}[Per-state attacker rewards, $f_i$]\label{def:fis}
	Let $f_i$ denote the expected reward of a canonicalized attacker block mined in \texttt{State i}.     
\end{definition}
To calculate this value, we need to find the expected value of the reward function by integrating the time distribution over the possible paths that include an attacker block claiming rewards arriving during \texttt{State i}. 
We first enumerate all possible paths that result in a canonical attacker block from \stateinosp; we then integrate the reward function over each path. The following example demonstrates this technique, and we generalize it in \Cref{thm:attackgeq2}.

\begin{example}[\texttt{State 3} paths]\label{ex:state3paths}
	Consider the rewards arriving after the attacker has a lead of length three. These rewards can be canonicalized in four different ways:
	\begin{enumerate}
		\item the attacker finds the next block, extending their lead to four,
		\item the honest parties find the next block, then the attacker finds the subsequent,
		\item the honest parties find the next two blocks, causing the attacker to publish their hidden chain, and then the attacker finds the first block after publishing,
		\item the honest parties find the next two blocks, causing the attacker to publish their hidden chain, and then the honest parties find the first block after that.
	\end{enumerate}
\end{example}
We can succinctly represent these four outcomes using the strings, \texttt{A, HA, HHA, HHH}, where \texttt{H} \& \texttt{A} denote honest and attacker blocks, respectively.
This example prompts the definition of attacker paths.
\begin{definition}[Attacker paths]\label{def:attack-paths}
	Given \texttt{State i} for all $i \geq 2$, there are $i$ distinct paths resulting in the attacker capturing rewards accrued in that state. The paths are enumerated as the string \texttt{(H$^*$)A}, where \texttt{H} \& \texttt{A} denote honest and attacker blocks respectively and \texttt{H} is repeated $0,1,\ldots i-1$ times.
\end{definition}

Continuing our \texttt{State 3} example, we now calculate the expected reward from each attacker path; adding these together is precisely the value of interest, $f_3$.

\begin{example}[$f_3$ continued]
	Consider the three attacker paths of \texttt{State 3}: \texttt{A, HA, HHA}. These paths have lengths 1,2,3 and occur with probabilities $\alpha, (1-\alpha)\alpha, (1-\alpha)^2\alpha$, respectively. Thus, we calculate the expected reward as,
	\begin{small}
		\begin{align*}
			f_3 =& 
			\underbrace{\alpha\int_{0}^\infty \frac{e^{-t/(1-\lambda)}}{(1-\lambda)} \mathbb{E}_{r}[R(t)]\,dt}_{\texttt{A}} +
			\underbrace{(1-\alpha)\alpha \int_0^\infty \frac{t e^{-t/(1-\lambda)}}{(1-\lambda)^2} \mathbb{E}_{r}[R(t)]\,dt}_{\texttt{HA}} \\
			&\quad + 
			\underbrace{(1-\alpha)^2 \alpha \int_0^\infty \frac{t^2 e^{-t/(1-\lambda)}}{2(1-\lambda)^3} \mathbb{E}_{r}[R(t)]\,dt}_{\texttt{HHA}}
		\end{align*}
	\end{small}
\end{example}
Each of these expressions can be viewed as the product of three independent sources of randomness. The coefficients of the integrals are the probabilities of each path determined by the winning miner, which depends on $\vec{m}$. The first expression in the integrand is the PDF of the Erlang Distribution, which measures the sum of i.i.d. exponential random variables (all with rate $1/(1-\lambda)$) to determine the amount of time of the path, which depends on $\vec{t}$. The second expression in the integrand is the expected value over all remaining randomness, $r$, of the reward function at time $t$. See \Cref{app:general-f0-f1} for the remaining $f_i$ calculations. Combining the stationary distribution values, $p_i$, with the per-state expected rewards, $f_i$, we can calculate the full expected attacker reward. 
\begin{theorem}\label{def:fullreward}
	The attacker's expected reward is,
	\begin{align*}
		\text{ATTACKER REWARD} =f_0p_0 + f_1 p_1 + \alpha \sum_{i=2}^\infty f_i p_{i-1}.
	\end{align*}
\end{theorem}
\begin{proof}
	For \texttt{State 0} and \texttt{State 1}, we multiply the stationary distribution probability by the expected per-state attacker reward to calculate the contribution to the full attacker reward. For \stateinosp, $i\geq 2$, we need to avoid double counting the contributions from each state (e.g., you can transition to \texttt{State 3} from either \texttt{State 2} \textit{or} \texttt{State 4}). To account for this we only consider the probability of arriving in each state from the $i-1$ state, which occurs with probability $\alpha p_{i-1}$. Thus, for each state, we add the contribution to the total attacker reward as $\alpha f_i p_{i-1}$. The resulting value tells us the expected attacker reward per unit time of following a $\beta$-cutoff strategy under the static reward function and as a function of $\alpha,\beta,\gamma.$ 
\end{proof}

\section{Selfish mining with three reward sources}\label{sec:multiplerewards}
Selfish mining strategies were analyzed with \textit{just} transaction fees and \textit{just} block rewards in \cite{eyal2013majority,carlsten2016instability}, respectively.
With the more general notion of miner rewards as defined in \Cref{sec:prelims}, a similarly general analysis is required to describe the profitability of selfish mining under different reward schedules.
The methodology of path counting and integrating the general reward function established in \Cref{sec:generalstatic} works for any static reward functions. We now instantiate a specific aggregate reward function, which more accurately captures complete miner incentives as they exist in Bitcoin today. 
This combined reward function, which we denote $\hat{R}$, is composed of (1) a fixed block reward of size $C$, (2) a linear-in-time transaction fee reward, and (3) an ``extra'' reward of size $E$ awarded to a block based on the outcome of a Bernoulli trial with probability $p$. Note that this new reward function considers the sum of each of these rewards, a more representative model of how miners are rewarded in reality rather than considering each of the rewards in isolation. For more straightforward examples of applying the path-counting technique to single-source reward functions, see \Cref{app:block-rews-only} for only considering block rewards as in \cite{eyal2013majority} and \Cref{app:lin-rews-only} for only considering transaction fees as in \cite{carlsten2016instability}.

\subsection{Rewards \#1 \& \#2: block rewards and transaction fees}
Each block that a miner produces earns a ``fixed block reward'' of magnitude $C$, which is paid directly to the miner as the first transaction in a block. We consider the block reward fixed.\footnote{The Bitcoin block reward is cut in half every four years, which impacts the relative size of the block reward compared to other reward sources. Our model considers the strategies available to miners within the same block reward period.}
\begin{remark}[Block rewards are static]
	Block rewards are a constant function that doesn't depend on $t$,
	\begin{align}\label{eq:blockrews}
		R(t) = C.
	\end{align}
	As such, they are static because each block reward is identically distributed no matter the timestamp of the parent block.
\end{remark}
The miners are also paid through the contents of the block they create. In particular, the transactions themselves specify a fee\footnote{In Bitcoin, the UTXO model defines a set of inputs and outputs for a transaction. Any balance that doesn't specify an output is claimable by the miner.} to be paid to the miner for including the transaction in the block. As in \cite{carlsten2016instability}, we start by assuming transaction fees arrive at a deterministic rate and are fully claimable by any subsequent block. See \cite{tsabary2018gap} for empirical measurements justifying the linear-in-time transaction fee rewards. 
\begin{remark}[Deterministic transaction fees with fully claiming blocks are static]
	Using the \cite{carlsten2016instability} definition of fixed-rate transaction fee arrival, we have
	\begin{align}\label{eq:txnrews}
		R(t) = t.
	\end{align}
	This reward is static, as it is deterministic and the same for all blocks (only depending on timestamp of the parent block).
\end{remark}

\subsection{Rewards \#3: non-deterministic extra rewards}
We also introduce a third type of reward to our model, motivated by the reality that some blocks have much higher transaction fee revenue than others due to contention. \cite{ZurAET23} use a similar model to capture high-fee-paying transactions in addition to block rewards; see \Cref{subsec:relwork} for further discussion. Consider, for example, that a new type of transaction can become available at a specific block height, and only a fixed amount of those transactions are valid (e.g., the first 10,000 transactions that purchase a specific NFT). To get their transaction included, participants submit bids specifying the fee they will pay to the block producer for higher-priority inclusion (assuming transactions are ordered by fee). This contention for block space leads to much higher revenue for the miner (who serves as the auctioneer) because even assuming infinite block sizes, the finite nature of the transaction type induces the competition (sometimes referred to as a ``priority gas auction'' \cite{daian2019flash}). We model this reward as a fixed size ``extra reward'' of magnitude $E$ available to a miner of a block with probability $p$ (a Bernoulli trial) and independent of time. We refer to this reward function as ``Bernoulli rewards.'' 

\begin{remark}[Bernoulli rewards are static]
	Bernoulli rewards are static because each block has the same distribution of rewards according to the outcome of the trial,
	\begin{align}\label{eq:bernoullirews}
		R(t) = \begin{cases}
			E & \text{if } X=1\\
			0 & \text{otherwise},
		\end{cases}
		\quad \text{where } X \sim \text{Bernoulli}(p).
	\end{align}
\end{remark}

Note that this model doesn't allow for the ``predictability'' of these Bernoulli rewards. Since miners may know a priori what block height a new set of transactions will arrive at, miners' strategy space would be different than the standard selfish mining strategies we explore below. See \Cref{sec:conclusion} for more discussion.

\begin{definition}[Reward function instantiation, $\hat{R}$]
	Combining the three reward sources (\Cref{eq:blockrews,eq:txnrews,eq:bernoullirews}), we have the full reward function, which we denote as $\hat{R}$,
	\begin{align}\label{eq:fullrews}
		\hat{R}(t) &= C + t + E \cdot \mathds{1}[X=1], \; X\sim \text{Bernoulli}(p).
	\end{align}
\end{definition}
Recall that the path-counting technique defined in \Cref{sec:generalstatic} applies to any static reward function. Since $\hat{R}$ is the sum of three independent, static rewards sources, it is static itself, and thus, we can analyze it. Under $\hat{R}$, we seek to calculate the attacker reward (\Cref{def:fullreward}). Following the structure above, we define the Markov Chain as a function of $\hat{R}$, which induces a stationary distribution $p_i$ before explicitly calculating the per-state attacker reward $f_i$. 
For the derivation of the stationary distribution, an instantiation of the general technique in \Cref{subsec:strategies}, and the Markov chain under this combined reward function, see \Cref{app:transition-probs-instance}. For the derivation of the expected attacker rewards, an instantiation of the general technique in \Cref{subsec:perstateattacker}, see \Cref{app:expected-rewards-instance}.

\subsection{Numerical results and discussion}\label{subsec:numerical}

\begin{figure}
	\centering
	\includegraphics[width=\linewidth]{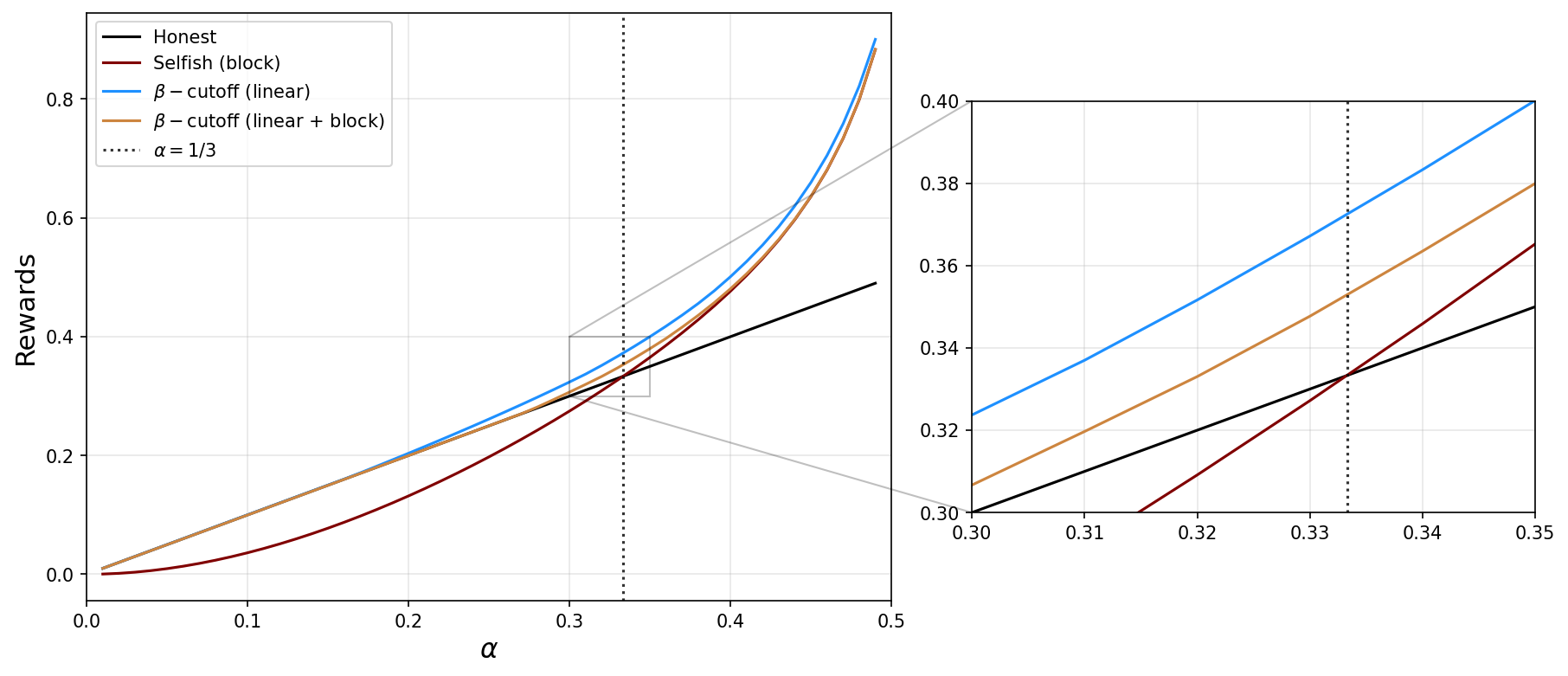}
	\caption{The attacker rewards as a function of $\alpha$ under different metrics of rewards. We consider miners who optimize for block rewards, linear-in-time rewards, or a combination of both. Considering both rewards together paints a more realistic picture of the protocol risk.}
	\label{fig:interpolation}
\end{figure}
\paragraph*{Linear-in-time transaction fees and block rewards.}
\Cref{fig:interpolation} analyzes the simple combination of the linear-in-time transaction fee rewards and block rewards. As expected, the attacker rewards under this combined function interpolates between the two extremes. \texttt{Selfish} (in red) shows the percentage of the block rewards collected when always hiding in \texttt{State 0} (which is exactly the reward in \cite{eyal2013majority} – see \Cref{app:block-rews-only} for the full derivation). \texttt{$\beta-$cutoff (linear)} (in blue) shows the percentage of the linear-in-time transaction fees collected on the attacker chain when choosing $\beta$ to maximize this ratio (which is exactly the reward in \cite{carlsten2016instability} – see \Cref{app:lin-rews-only} for the full derivation). \texttt{$\beta-$cutoff (linear + block)} (in tan) shows the attacker's reward when considering both reward sources together. 
One interpretation of \Cref{fig:interpolation} examines how different reward regimes can lead to dramatically different conclusions regarding the ``risk of attack'' a protocol faces. 
In this case, the selfish miner who only optimizes for the ratio of block rewards is not profitable until $\alpha=1/3$. On the other hand, if we only consider the fraction of linear-in-time transaction fees capturable by a $\beta$-cutoff selfish miner, the story looks much worse. 
In particular, that miner becomes profitable around $\alpha=0.15$. Considering both rewards results in a more measured conclusion, where the strategy becomes profitable around $\alpha=0.25$. 
By varying the relative size of the block reward compared to the per-unit linear-in-time transaction fees, we can thus fully capture the dynamics of both reward models by interpolating between the two strategies, which consider the sub-rewards in isolation. Additionally, this figure can be interpreted qualitatively. We see that the attacker considering both rewards (tan) behaves \textit{less aggressively} than the linear optimizing attacker (blue) for $\alpha \in [0.15, 0.25]$, as the optimal reward in that range is equivalent to honest. Conversely, for $\alpha \in [0.3, 0.33]$, a pure selfish mining strategy would not be profitable; thus, the attacker considering both rewards would be \textit{more aggressive} than the block-reward maximizing miner (who would choose to mine honestly).

\begin{figure}
	\centering    
	\includegraphics[width=0.63\linewidth]{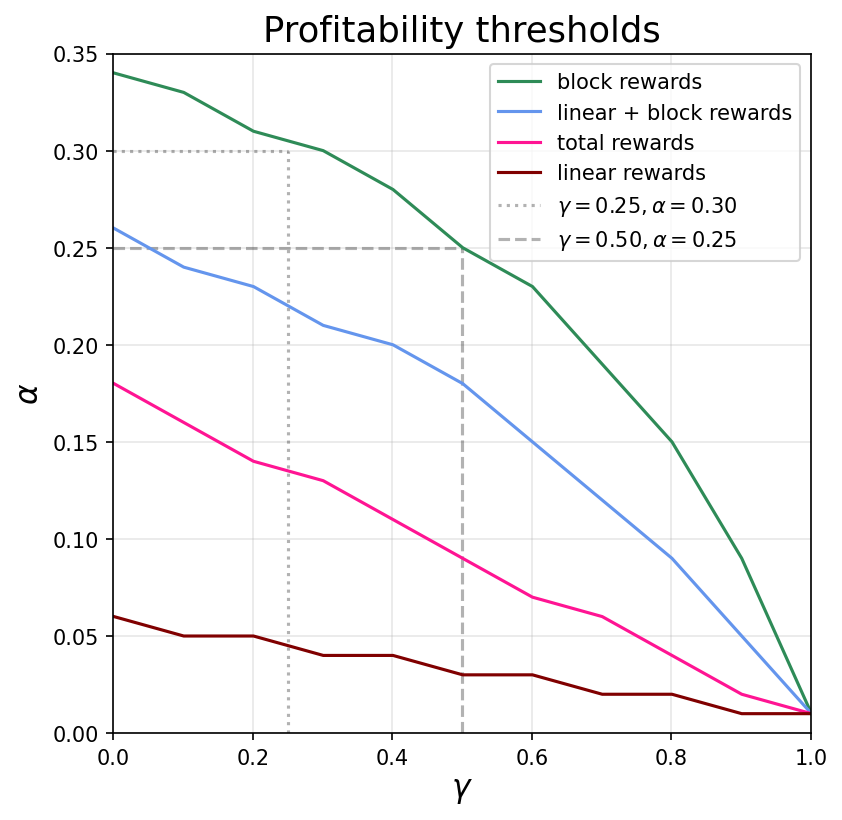}
	\caption{Demonstrating the $\alpha$ at which each strategy becomes profitable over honest as a function of $\gamma$. This extends Figure~3 from \cite{eyal2013majority} to include more strategies. Each respective strategy considers profitability when only measuring a subset of the total rewards. For example, \texttt{linear + block rewards} (in blue) denotes a $\beta-$cutoff strategy for $\alpha$ profitable if, when selecting $\beta$ to maximize the sum of linear and block rewards, the expected attacker reward exceeds $2\alpha$.}
	\label{fig:threshold-alphas}
\end{figure}

\paragraph*{Profitability thresholds.}
\Cref{fig:threshold-alphas} shows the value of $\alpha$ at which various strategies become profitable under different reward sources as a function of $\gamma$. This extends Figure~3 of \cite{eyal2013majority} to include more strategies. For each $\gamma$, we consider the optimal $\beta$ cutoff for an attacker, maximizing block, linear, and total rewards, respectively. For each candidate $\alpha$, we check if the optimal $\beta$ results in a total reward that exceeds the benchmark of the honest performance under that reward function (i.e., the proportional block rewards from honest mining). We find the lowest candidate $\alpha$ such that the rewards exceed the benchmark and identify that as the profitability threshold. Intuitively, this is the fraction of the mining power needed to perform this strategy profitably. 

For the pure selfish miner (in green), we see that the profitability thresholds of $1/3,0.3, 0.25$ for $\gamma=0,0.25,0.5$ are identical to \cite{eyal2013majority}. When considering just linear and block rewards (in blue) and the total rewards (linear + block + bernoulli) (in pink), we see that for all values of $\gamma$, the profitability threshold decreases significantly. For example, at $\gamma=0$, the profitability threshold is reduced from $1/3 \rightarrow 0.26 \rightarrow 0.18$ (reductions of $22\%$ and $31\%$ respectively) when considering the different reward sources. Similarly, at $\gamma=0.5$, the profitability threshold is reduced from $0.25 \rightarrow 0.18 \rightarrow 0.09$ (reductions of $28\%$ and $50\%$ respectively).

Interestingly, the attacker that only considers linear-in-time transaction fees (shown in red) is profitable (measured against just linear-in-time rewards) for nearly all values of $\alpha$. While this may seem concerning, we believe that only linear-in-time rewards are not a large enough component of the current Bitcoin incentive structure to warrant concern. The aggregate view of the rewards (e.g., total shown in pink) more accurately represents rewards as they exist in Bitcoin today. These results only consider the case where all three components are approximately equal in magnitude. See \Cref{sec:conclusion} for a discussion on how choosing the relative sizes of these different reward sources based on empirical values would be a valuable application of our technique.

\paragraph*{Measuring Bernoulli reward.}
\begin{figure}
	\centering
	\includegraphics[width=\linewidth]{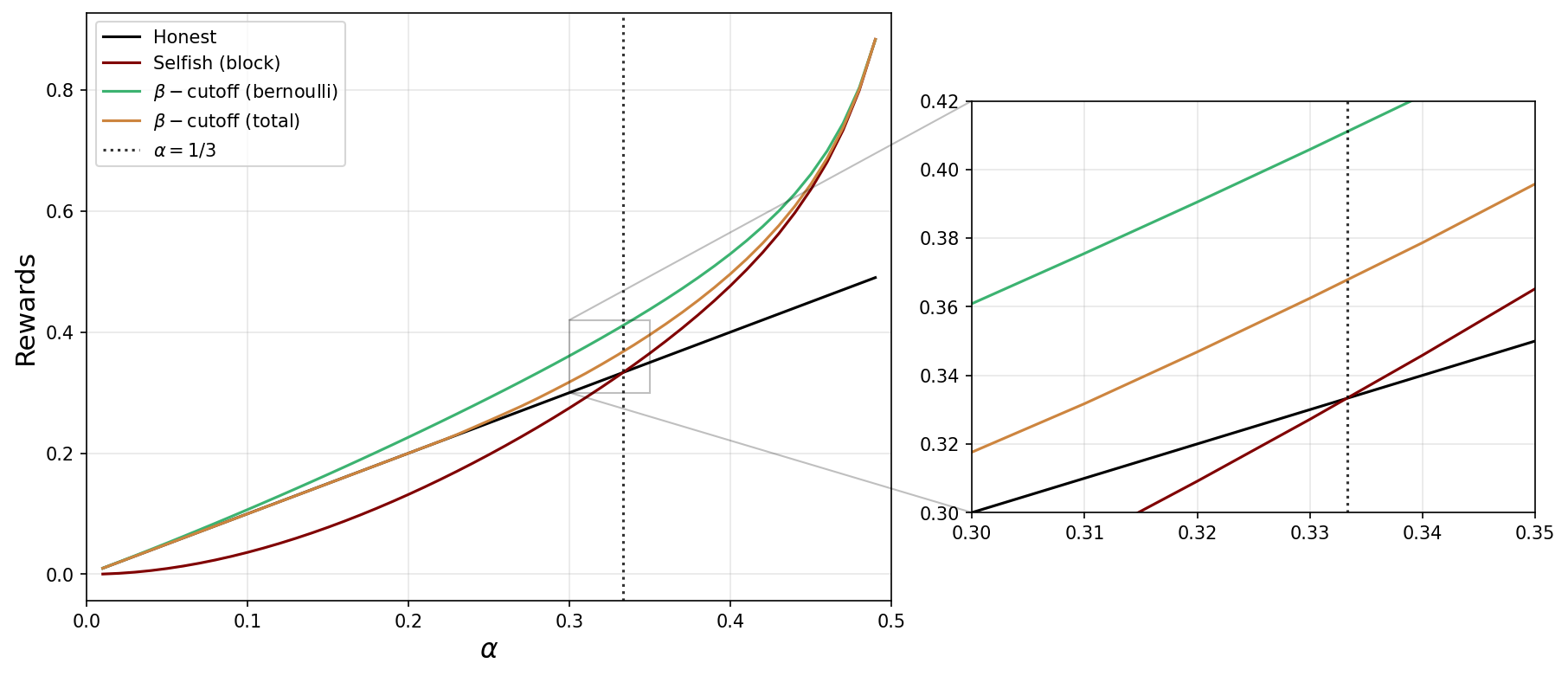}
	\caption{The attacker rewards as a function of $\alpha$ under different metrics of rewards. We consider miners who optimize for block rewards, Bernoulli rewards, and the full $\hat{R}$ containing block, Bernoulli, and linear rewards. Note that the Bernoulli reward-optimizing attacker is profitable for all values of $\alpha$ and meaningfully deviates from honest for $\alpha>0.1$.}
	\label{fig:bernoullis}
\end{figure}
\Cref{fig:bernoullis} examines the profitability of two other mining strategies: optimizing $\beta$ for Bernoulli rewards (in green) versus optimizing $\beta$ for the sum of linear, block, and Bernoulli rewards (in tan). Again, the combined rewards interpolate between the Bernoulli and the block-optimizing miners. For the miner maximizing over all three rewards, we normalize them each to have an expected value of $1$ per block. The miner who only considers Bernoulli rewards (in green) is always profitable and significantly outperforms honest when $\alpha\geq 0.1$. 

Bernoulli rewards aim to capture the volatility of MEV, which depends greatly on exogenous randomnesss (e.g., global capital market volatility). As the block reward continues on it's halving schedule and if transaction fees for Bitcoin transfers remain relatively low, MEV may come to dominate the miner incentives in the long-term. The fact that a miner becomes profitable under MEV rewards at such a low value of $\alpha$ may be a canary in the coal mine for the endgame analysis of selfish mining under MEV. We see understanding how Bitcoin MEV may play out under increased financial activity secured by the chain (e.g., through Bitcoin L2s that need to post data to the base chain itself) as one of the most important avenues for future work. 

See \Cref{app:extended-figures} for more numerical examples and simulation results confirming the accuracy of the technique.

\section{Conclusion and future work}\label{sec:conclusion}
We hope this work is a starting point for a more complete picture of participants' incentives in permissionless consensus mechanisms. \textbf{Our practical contributions of analyzing selfish mining under a more representative reward model in \Cref{sec:multiplerewards}} fills a long-standing gap in the selfish mining literature. \textbf{Our modeling contribution of general rewards in the NCG and characterizing properties of observed reward sources in \Cref{sec:prelims,subsec:properties}} serves as the basis for a richer theoretical treatment of MEV and its implications for consensus. \textbf{Our methodological contribution of presenting a path-counting technique to calculate expected attacker profit under general static rewards in \Cref{sec:generalstatic}} provides a turn-key solution for future work to analyze selfish mining under newly identified reward sources.  
To conclude, we outline many potential future research directions.

\emph{Applying our methodology more broadly.} Our reward instantiation in \Cref{sec:multiplerewards} is a reasonably realistic model of reward sources in the Bitcoin blockchain today.
The methodology and instantiation represent a significant step in understanding the risk of selfish mining in the presence of multi-faceted rewards, especially since prior work generally considered one reward source at a time. 
However, empirical analysis may strengthen our results by forming a more nuanced understanding of these rewards in practice (e.g., measuring the relative size and probability of different MEV events). Note that our methodology still applies to any static reward sources that can be analytically calculated using the path-counting technique presented in \Cref{sec:generalstatic}. 
Beyond explicitly using our methodology, our technique has relatively straightforward extensions that can reach beyond static rewards and $\beta-$cutoff strategies.

\emph{Extending our methodology.} There are several natural extensions to our methods. For example, considering the profitability of $\beta$-cutoff selfish mining under non-static reward functions is feasible. Such reward functions depend on additional information not captured in the states of the Markov Chain. 
However, suppose that the additional information is exogenous to the chain and independent of views. 
In that case, it is possible to augment the state space of the Markov Chain to include this information.
Sme reward sources exhibit ``periodic'' behavior; in the case of Babylon (\Cref{ex:babylon}), elevated rewards persisted for a seven-block period. A simple modification of the Markov chain would allow the rewards to depend on whether the system was in a ``high'' versus ``low'' regime.

Another extension is to study MDP-based optimal strategies as in \cite{sapirshtein2017optimal} rather than $\beta$-cutoff selfish mining. \cite{ZurAET23} demonstrate the impact of changing the reward function on optimal selfish mining profits when considering the combination of block rewards and occasional ``whale'' transactions, and they note that the resulting large state spaces were intractable with traditional MDP solving tooling and required machine learning. Considering how to more succinctly represent multi-reward state spaces or using the Deep RL approach with more combinatorial rewards are both promising directions. While we consider strategies that make decisions based on the realization of rewards in the \emph{current} block, the broader MDP strategy space can be future-looking; for example, an attacker may want to start creating a hidden chain of several blocks in advance of an anticipated large reward (e.g., from an NFT drop occurring at a specific block height). 


\emph{A complete picture of consensus incentives.}
As demonstrated in \Cref{ex:babylon,ex:lowcarb}, modern blockchains have faced and will continue to face distortion of consensus incentives from the application layer handling growing amounts of economic activity. \Cref{subsec:properties} is a first step at modeling properties of general reward functions, but applying these properties to MEV beyond the transaction fee (\Cref{sec:examples}) and LVR (\Cref{app:lvr-examples}) case studies remains open. It would be interesting to derive a set of necessary and sufficient properties to fully taxonomize MEV. Studying heterogeneity of reward sources was out of the scope of this work, but expanding the properties of reward functions when block producers have highly different rewards realizations is another key future direction.


	\bibliography{aft}
	
	\newpage
	\appendix

\section{Extended examples}\label{app:extended-examples}

\subsection{Timing games}\label{app:timing-games}

\begin{example}[Timing games]\label{ex:timinggames}
In Proof-of-Stake protocols, no random mining process dictates the progression of time. Instead, time is explicitly discretized, and the protocol elects a leader as the sole block producer for a given slot. As in Proof-of-Work, stakers who produce valid blocks are compensated with new tokens (issuance)
– a protocol-prescribed consensus reward. For a proposed block to be accepted by the remainder of the network, it must arrive at the other nodes by a deadline. Typically, the protocol specifies that the proposer releases the block relatively early to ensure the rest of the network has plenty of time to receive it before deciding which chain to extend (e.g., in the Ethereum protocol, there is a four-second delay between the expected publication time and when the next voters determine whether the block was available or not). If the consensus rewards fully captured the incentives of stakers, the proposer would never delay their block publication, as any delay would increase the risk of the block not being received due to network latency. Yet  \cite{schwarz2023time} and \cite{oz2023time} model and measure the increase in rewards for intentionally delaying the publication of a block, a phenomenon referred to as ``timing games.'' Here again, application layer rewards distort the overall incentives of the game. Proposers benefit from the fact that any additional time allows for increased transaction fees and MEV to accrue. Thus, in some cases, delaying their block and risking losing the entire reward may be worth waiting extra time. 
\end{example}

\subsection{LVR emphemerality}\label{app:lvr-emphermeral}
\begin{example}[LVR is ephemeral in Proof-of-Work]\label{ex:lvr-ephemeral}
    LVR, as presented in \cite{lvr}, measures the profit of arbitrageurs who are instantaneously balancing the price of a decentralized exchange (abbr. DEX) with an infinitely deep centralized exchange (abbr. CEX). In other words, the profits depend on constantly executing trades on both venues to ensure the DEX price matches the CEX. In leader-election protocols like Proof-of-Stake, this might be reasonable. Once a leader is known, they start performing the trades and can be certain that the block they produce will contain each of those trades and become part of the canonical chain. When the next block producer is uncertain, as in Proof-of-Work, this model breaks down. Miners don't know they will produce a block a priori; thus, they will not execute trades on the CEX while mining. Instead, a more reasonable strategy is to perform the DEX leg of the arbitrage a single time as the first transaction in their block once they mine it and only then execute the CEX leg to complete the arbitrage. This distinction is critical. In the \cite{lvr} model, LVR is monotone increasing and accumulating for the block production period. In Proof-of-Work, the price on the CEX could retrace by the time a block is mined, eliminating the arbitrage profit that may have been present earlier. This ``ephemerality'' (and its inverse ``persistence'' \Cref{def:persistent}) is one property of rewards that we capture in our framework. 
\end{example}

\section{Further properties of reward functions}\label{app:max-blocks-persistent}
\begin{definition}[Maximum Rewards \& Maximizing Blocks]\label{def:max-block}
    Given a reward function $R$, we define the maximizing block function $B_{\text{opt}}$ as 
    \[
    	B_{\text{opt}}(t,V,B,r):=\underset{B'\in\mathcal{B}(t,V,B,r)}{\arg\max} R(t,V,B,r,B').
    \]
    We further define the maximum reward function $R_{\text{opt}}$ as 
    \[
    	R_{\text{opt}}(t,V,B,r):=R(t,V,B,r,B')
    \] for some $B'\in B_{\text{opt}}(t,V,B,r)$.
\end{definition}

Observe that if a reward function $R$ is static, then $\Ropt(t,V,B,r,B')$ can be rewritten as a two-variable function of just $r$ and the time $\Delta$ between $\tstamp(B)$ and $t$. 

We next define \emph{persistent rewards}, which arrive at some time and can be claimed at most once. Upon arrival, they remain indefinitely claimable by any block whose ancestors have not already claimed them. Let $\Rc(B)$ denote the amount of reward attributed to the block creator if the block becomes canonical and $\chain(B)$ the set of blocks on the ancestral path of $B$ (including $B$).

\begin{definition}[Persistent Rewards]\label{def:persistent}
A reward function $R$ is \emph{persistent} if for all realizations of $\vec{t},\vec{m},r$, at any time $t$, for all blocks $B$ in the resulting view $V$, we have:
\begin{itemize}
    \item for all $B'\in\mathcal{B}(t,V,B,r)$, 
    \begin{align}\label{eq:persistence}
        R(t,V,B,r,B')\leq R_{\text{opt}}(t,V_0,B_0,r)-\sum_{B''\in \chain(B)} \Rc(B''),
    \end{align}
    \item there exists some $B'\in\mathcal{B}(t,V,B,r)$ for which the above holds with equality.
\end{itemize}
\end{definition}

We sometimes call a non-persistent reward function \emph{ephemeral}. \Cref{rem:patient-persistent} highlights that transaction fees are persistent if the users creating the transactions are patient (willing to wait for inclusion and not cancel pending transactions). On the other hand, fees from transactions submitted by impatient users (as in \Cref{rem:impatient}) are not persistent since the canceled transactions are no longer claimable by future blocks. 

Persistent rewards are not affected by orphan or uncle blocks, but they \emph{may} be view-dependent since they are affected by the claimed rewards on the ancestral path of a block. The following lemma states that persistent rewards functions are view-independent if all blocks in all valid views claim the maximum available rewards.

\begin{lemma}[Persistent \& Maximizing Blocks $\implies$ View-Independent]\label{lem:all-claim}
    Let $R$ be persistent. Then $R$ is view-independent if for all $t$, all $V\in\mathcal{V}_t$, all parent-child blocks $B,B'$ in $V$, and all $r$, we have $\Rc(B')=\Bopt(t,V,B,r)$.
\end{lemma}

\begin{proof}
    Suppose there is a view $V$ at time $t$ in which some blocks do not claim all rewards. Let $B_1$ mined at $t'$ be the earliest such block. Consider the prefix of $V$ as of time $t'$, and call it $V_1$. Let $B^*\in\Bopt(t,V_1,B_1,r)$ be the reward-maximizing block extending $B_1$. 
    
    Now consider a different view $V_2\in\mathcal{V}_{t'}$ that is identical to $V_1$, except $B_1$ is replaced with a reward-maximizing block $B_2\in\Bopt(t',V_{t'},\text{parent}(B_1),r)$. By persistence applied to $V_2$, we have
    \begin{align*}
        R(t,V_2,B_2,r,B^*) \leq& \Ropt(t,V_0,B_0,r)-\sum_{B''\in\chain_{V_2}(B_2)}\Rc_{V_2}(B'') \\
        =& \Ropt(t,V_0,B_0,r)-\Rc_{V_2}(B_2) -\sum_{B''\in\chain_{V_2}(\text{parent}(B_2))}\Rc_{V_2}(B'') \\
        =& \Ropt(t,V_0,B_0,r)-\Rc_{V_1}(B_1) -\sum_{B''\in\chain_{V_1}(\text{parent}(B_1))}\Rc_{V_1}(B'') \\
        <& \Ropt(t,V_0,B_0,r)-\Rc_{V_1}(B_1)-\sum_{B''\in\chain_{V_1}(B_1)}\Rc_{V_1}(B'') \\
        =& R(t,V_1,B_1,r,B^*).
    \end{align*}
    where the first inequality follows from persistence applied to $V_2$, the next equality is algebra, the next equality is by construction of $V_2$, the next inequality is by assumption that $B_1$ is not reward-maximizing and $B_2$ is, and the last equality is from persistence applied to $V_1$, and in particular invoking the second bullet in the definition of persistence. 

    This is a contradiction, since by view-independence, $B^*$ should have the same reward in $V_1$ and $V_2$.
\end{proof}

    
    

The following lemma shows that static and persistent reward functions accrue linearly over time since the parent block, with a constant slope and intercept across blocks (but may be random depending on $r$). If a reward function is persistent and static, it can be simulated by drawing the randomness of $r$ to set the slope $a$ and the intercept $b$ of the maximum available reward function $\Ropt$. Then, for any block $B$ in any view, the reward for extending $B$ at time $\tstamp(B)+\Delta$ equals $a\cdot \Delta + b$. This is the model of transaction fee accrual in \cite{carlsten2016instability} and MEV accrual in \cite{schwarz2023time}.

\begin{lemma}[Static \& Persistent $\implies$ Linear]\label{lem:linear}
    Let $R$ be static and persistent. If $\Ropt(t,V,B,r)$ is differentiable with respect to $t$, then it is of the form $a(r)\cdot (t-\tstamp(B)) + b(r)$.
\end{lemma}
\begin{proof}
Consider a static and persistent reward function $R$. Since $R$ is static, it is view-independent, so by \Cref{lem:all-claim}, we can restrict attention to views in which all blocks claim all rewards.

Consider the function $\Ropt$. Since $R$ is view-independent (because it is static), we can drop its dependence on $V$. Since $R$ is static, its only dependence on $B$ is through $t-\ts(B)$. Therefore, $\Ropt$ can be written as a function $f(\Delta,r)$, where $\Delta$ is the time since the creation of the parent block. We must show that $f$ takes the form $a(r)\cdot \Delta+ b(r)$.

Fix times $t'<t$ and $r$, Consider a view at time $t$ consisting of three blocks. The genesis block $B_0$, with a child $B'$ mined at $t'$, and grandchild $B$ mined at $t$. 
\begin{align*}
    f(t,r)&=f(t-t',r)+\Rc(t',r) \tag{persistence applies to $B$}\\
    &=f(t-t',r)+f(t',r) \tag{$B'$ claims all rewards}
\end{align*}
Rearranging, dividing by $t-t'$, we get
\[
\frac{f(t,r)-f(t',r)}{t-t'}=\frac{f(t-t')}{t-t'}
\]
Taking the limit $t'\to t$ (which exists since $\Ropt$ is differentiable), the left-hand-side is equal to $d/d\Delta f(t,r)$, while the right-hand-side is equal to $d/d\Delta f(0,r)$. Since the choice of $t$ was arbitrary, we conclude that the derivative of $f$ with respect to $\Delta$ is a function of $r$ and constant for all $\Delta$.

\end{proof}

Note that the transaction fees defined in \cite{carlsten2016instability} are linear; we use this same reward function as part of our instantiation in \Cref{sec:multiplerewards}.

\subsection{Persistent reward examples}\label{app:persistent-reward-examples}

\begin{example}[Patient transaction fees with infinite capacity blocks are persistent]\label{rem:patient-persistent}
    The reward function of a candidate block $B'$ built upon a parent block $B$ is bounded above by the sum of transaction fees not claimed by any block in $\chain(B).$ 
    For any parent block $B$, the block $B'$ that contains all transactions not included in $\chain(B)$ is valid (because users are patient and blocks have infinite size) and satisfies the equality in \Cref{eq:persistence}.
\end{example}
Transaction fees cannot be persistent without infinite capacity blocks because equality will not hold if the block cannot fit all available transactions.

\begin{example}[Bernoulli rewards are not persistent]\label{rem:bernoulli-ephemeral}
    The reward function is the outcome of the Bernoulli trial and does not allow for previous iterations of the trial to be captured in the same block (only one reward per block à la block rewards). This violates the equality condition of \Cref{eq:persistence} and is not persistent.
\end{example}

Patience levels have been studied in the context of transaction fees \cite{nisan2023serial,penna2024serial,babaioff2024optimality}. 
In practice, rewards might persist over some time but not indefinitely. For example, users might have limited patience of a few blocks rather than being fully patient (\Cref{rem:patient-persistent}) or fully impatient (\Cref{rem:impatient}). Other types of MEV may similarly only satisfy ``partial persistence.'' For example, sandwich attacks persist if the DEX price is within the slippage limit of the user's swap. A complete MEV taxonomy is out of scope for this work; see \Cref{sec:conclusion} for a discussion on natural modeling and empirical extensions.

\subsection{Properties of CEX-DEX Arbitrage}\label{app:lvr-examples}
Loss-Versus-Rebalancing (\Cref{ex:lvr-ephemeral}) measures the profits earned by the arbitrageurs who balance the price of a DEX against an infinitely deep CEX. The model of \cite{lvr} assumes that the arbitrageurs continuously trade as the CEX price moves according to a Geometric Brownian Motion (abbr. GBM) stochastic process. This price movement is external and independent of the randomness of the chain and thus is captured by $r$ in our model. 
While the LVR literature does not explicitly model consensus, the profits of these arbitrageurs can be viewed as a form of MEV. The block producer fully controls the on-chain leg of the arbitrage and can replicate the strategy by continuously trading on the CEX while also continuously updating the DEX price within their block.
In Proof-of-Work, this implies that \textit{all} miners are continuously executing trades on the CEX because the next block producer is unknown. 

\begin{example}[LVR is persistent if all miners continuously trade]\label{rem:lvr-persistent}
    All miners trading continuously implies that the CEX and DEX prices are aligned at every block. In any resulting view, a continuously trading miner that mines a block at time $t$ with a parent mined at $t'$ collects the total amount of LVR during the interval $[t',t],$ as per Equation (8) in \cite{lvr}. This reward function always satisfies the second bullet in \Cref{def:persistent} and is thus persistent.
\end{example}

LVR is only persistent if miners constantly trade without knowing a priori that they will mine the subsequent block. Additionally, blocks must have infinite capacity to include the complete set of DEX trades that the miner performs during the mining process. This is consistent with the literature on LVR and might be a reasonable assumption in a Proof-of-Stake protocol where the block producer knows that they have the right to produce a block at an assigned time (e.g., in Ethereum, where the schedule of the following 64 block producers, about 10 minutes worth, is public information \cite{ethereum_consensus_specs_compute_proposer_index}). 
In Proof-of-Work, however, this model of LVR may not be a reasonable assumption as only a single miner will realize the profit from the arbitrage. The miners that lose the race execute only the CEX trades without the corresponding DEX leg of the arbitrage. Performing only the CEX trades \emph{loses} money in expectation. If the CEX price moves up from $p \nearrow p'$, the CEX leg of the arbitrage sells low (marked to the more recent and thus fair price $p'$). The same logic holds when the price moved down from $p \searrow p'$, resulting in the CEX leg buying high.

For this reason, strategic miners would instead perform a ``discrete'' version of the trade, performing the arbitrage only once to align the DEX price to the CEX at the moment of block production.\footnote{For Bitcoin specifically, the ten-minute block times make it unlikely to see significant DEX trading volumes. Discrete LVR is still the correct model for consensus protocols where block producers face uncertainty about whether they will successfully produce the next block (such as DAG consensus and Proof-of-Work with faster block times).}
We refer to this as ``discrete LVR'' because both legs happen simultaneously upon block creation rather than continuously during mining. 
\begin{example}[Discrete LVR is not persistent]\label{rem:discrete-lvr-not-persist}
    Consider a block mined at time $t$ with a parent mined at $t'$. The discrete LVR reward function captures the arbitrage profit from balancing the DEX to a CEX price a single time based on the price movement on the CEX in $[t',t]$. This is \textit{not} persistent. 
    Consider a parent-child pair of blocks $B_1,B_2$ mined at time $t_1 < t_2$ when the CEX price is $p_1 < p_2$ respectively. Assume that DEX and CEX prices are aligned in $B_1$ and $B_2$, and in particular, note that $B_2$ receives a positive discrete LVR reward. Now suppose that at time $t > t_2$, the CEX price retraces back to $p_1$. The maximizing block $B$ that extends $B_1$ at time $t$ has a discrete LVR reward of 0 because the prices on the CEX and DEX match at $t_1$ and $t$.
    However, both $B_2$ and the maximizing block $B$ extending $B_2$ at time $t$ have strictly positive discrete LVR rewards, violating
    \Cref{eq:persistence} in the definition of persistence.
\end{example}

Intuitively, discrete LVR is not persistent because the arbitrage profits can disappear if they are unclaimed at a specific time (just like impatient transaction fees in \Cref{rem:impatient}). The previous two examples characterized how LVR is persistent or ephemeral depending on the leader's advanced knowledge. The following two examples show that LVR is generally not static, except under some locality assumptions. In \Cref{sec:generalstatic}, we analyze the profitability of a selfish mining variant under general static rewards (in particular \Cref{rem:local-lvr-static} below).

We start with a closer examination of the LVR calculation in Equation (8) of \cite{lvr}, which defines LVR over a time interval as the integral of the \textit{instantaneous LVR}. Instantaneous LVR is a function of three variables: the price $P$ of the asset on the CEX, the standard deviation of the GBM representing CEX price movements, and the \emph{marginal liquidity} of the DEX at $P$ (denoted by $|x^{*'} (P)|$ in \cite{lvr}), which is a deterministic function of $P$. 
Observe that to calculate instantaneous LVR at time $t$, knowing the current price level is necessary and sufficient.
The sufficient direction implies LVR is view-independent, while the necessary direction implies LVR is not static. \Cref{rem:lvr-vi,rem:non-local-lvr} formalize this.

\begin{example}[LVR with per-block aligned CEX and DEX prices is view-independent]\label{rem:lvr-vi}
    Restrict the set of views to ones that fully align CEX and DEX prices at each block (e.g., through each miner collecting either discrete LVR as in \Cref{rem:discrete-lvr-not-persist} or continuous LVR as in \Cref{rem:lvr-persistent}).
    Then the LVR reward function for $B'$ extending either $B_1,B_2$ both with timestamp $t'$ in views $V_1,V_2$ respectively depends only on the timestamp of the parent block (and the corresponding CEX price at that time) and the random price movements of the CEX under $r$ after $t'$. Therefore, LVR in this setting is view-independent.
\end{example}

This view-independence arises from the CEX and DEX price alignment at each block, which is similar to the mempool clearing from infinite block sizes in \Cref{rem:fullyclaimingviewindep} and from user impatience in \Cref{rem:impatient}. In these examples, the reward function only depends on events occurring after the parent block is mined. 

\begin{example}[LVR is not static]\label{rem:non-local-lvr}
    No matter the restrictions we place on views and miner strategies, LVR cannot be static because the distribution of rewards depends on the price level, an exogenous variable that changes as a function of time. The reward function for LVR depends on the realized price movements on the CEX during the block creation process, which in turn depends on the price level at that time. 
\end{example}

This last example highlights a significant limitation of static rewards generally – static rewards cannot vary based on exogenous randomness. The same distinction is present in \Cref{rem:patientisstatic} and \Cref{rem:fullyclaimingviewindep}, where the distribution (over external randomness) of the reward function varying in time reduces reward sources from static to only view-independent. The methodology and analysis we present in \Cref{sec:generalstatic,sec:multiplerewards} focus on static rewards as these are capturable in a relatively simple Markov Chain. See \Cref{sec:conclusion} for a discussion on extending the state space of the Markov Chain to capture non-static rewards.

While LVR is not static, we introduce a different reward function that \emph{is} static and argue it approximates LVR within local price neighborhoods.

\begin{example}[Resetting LVR is static]\label{rem:local-lvr-static}
    Restrict the set of views to those where both CEX and DEX prices upon creation of each block are exactly $P$. Resetting LVR is the reward function that starts a new GBM at $P$ for each block and grows identically to LVR between blocks. Both continuous and discrete versions of resetting LVR are well-defined in this manner.
    The resetting-LVR reward function, in either case, is static since it depends on \emph{only} the CEX price movements under $r$ since the parent block 
    (and \emph{not} on the price level when the parent block was created). In particular, for all $t$ and all $\Delta$, the resetting-LVR reward for the maximizing block at time $t$ with a parent mined at $t-\Delta$ has the same distribution – that of LVR starting at price $P$ after time $\Delta$ has passed.
\end{example}

We claim that resetting LVR is a reasonable local approximation to LVR over a small time frame. During a short time interval, price movements are bounded, and so is the effect of changes in $P$ on instantaneous LVR.\footnote{We intentionally state these claims informally since the goal of these examples is to illustrate the applicability of the properties we introduce in \Cref{subsec:properties}. More formal versions are possible but would require a deeper dive into the particular math behind LVR, which is beyond the scope of this paper.}
To summarize, per-block price alignment implies view-independence of LVR as demonstrated in \Cref{rem:lvr-vi}. LVR is not static (\Cref{rem:non-local-lvr}) because it depends on the price level of the CEX as of the parent timestamp. Resetting LVR (\Cref{rem:local-lvr-static}) differs because the price resets each block, removing the dependence on the parent timestamp (with the only remaining dependence being on \textit{time since} the parent block), making it static. 

\section{Extended model details}\label{app:model-details}

\subsection{Nakamoto Consensus Game algorithm}\label{app:algo}
\begin{algorithm}
\caption{View evolution under the Nakamoto Consensus Game}
\begin{algorithmic}[1]
    \STATE Draw independent random variables: $\vec{m}, \vec{t}, r$.
    \STATE Set $V^m_0=V_0$ for all $m\in M$, where $V_0:=\{B_0\}$ and $B_0$ is genesis.
    \STATE Set $t=0, j=1$ (let $t_j$, $m_j$ denote the time and miner of the $j^{th}$ block).
    \WHILE{game continues}
        \STATE Set $m = \nxtb{t,r}$ as the next miner scheduled to broadcast.
        \STATE Set $t' = \nxt{m,t,V_t^m,r}$ as the next broadcast time.
        \IF{$t' \leq t_j$}
            \STATE Update all views to reflect $m$'s newly broadcast blocks.
            \STATE Set $t \leftarrow t'$.
        \ELSE
            \STATE Examine the strategy of $m_j$ with inputs:
            \STATE \quad $t_j$, $V_t^{m_j}$, and $R^{m_j}(t_j, V_{t_j}^{m_j}, B, r, B')$ 
            \STATE \quad for all $B \in V_{t_j}^{m_j}$ and all $B' \in \mathcal{B}^{m_j}(t, V_{t_j}^{m_j}, B, r)$\\
            \STATE \quad to determine the parent and contents of $m_j$'s new block.
            \STATE Update $V_{t_j}^{m_j}$ to include this block.
            \STATE Set $t \leftarrow t_j$.
            \STATE Increment $j \leftarrow j + 1$.
        \ENDIF
    \ENDWHILE
\end{algorithmic}
\label{alg:ncg-updates}
\end{algorithm}

\paragraph*{Ensuring uniqueness of the longest chain}
We modify the NCG defined in \cite{bahrani2024undetectable} to account for more general reward functions.
Each miner $m$ collects the sum of rewards claimed in its blocks on the eventual longest chain and has utility proportional to the amount of reward it collects per unit of time. Formally, a \emph{longest chain} at time $t$ is any block in $V_t$ of greatest height. If the longest chain at time $t$ is unique, we denote by $\text{REWARD}_m^t$ the sum of rewards claimed by blocks mined by $m$ in the longest chain. If the longest chain at time $t$ is not unique, we let $t'<t$ denote the most recent time when the longest chain at $t'$ is unique, and define $\text{REWARD}_m^t:=\text{REWARD}_m^{t'}$. Recall that the longest chain at time 0 is unique by assumption, so this is always well-defined. The utility of miner $m$ is $\lim\inf_{t\rightarrow\infty} R_t^m/t.$ 

\subsection{Further model details}\label{app:further-model}
\paragraph*{Comparison to prior work.}
The majority of previous literature on selfish mining \cite{eyal2013majority,sapirshtein2017optimal,nayak2016stubborn} considers block rewards as the only source of revenue for miners and thus implicitly models difficulty adjustment by defining the miner utility in terms of the percentage of blocks on the longest chain. Maximizing this objective is equivalent to maximizing the per-unit-time profit because difficulty adjustment ensures the total amount of block rewards issued per unit of time is fixed. \cite{carlsten2016instability} considers transaction fees as the sole source of miner revenue. Similarly to block rewards, these transaction fees accrue at a fixed rate per unit of time and are assumed to remain claimable any time after arrival. In both cases, the sum of the rewards collected by honest and attacker blocks per unit of time remains constant.

In practice, many sources of miner revenue may vary over time. For example, \Cref{rem:non-local-lvr} (LVR) describes one source of revenue that grows super-linearly in inter-block time, and \Cref{rem:patient-not-vi} (transaction fees with finite blocks) describes another source that grows sub-linearly.
This means that, even if difficulty adjustment guarantees a fixed \emph{average} block time, the total reward collected by honest and attacker blocks depends on the specific inter-block times. 
Therefore, a profit-maximizing attacker would not simply maximize the \emph{percentage} of rewards they collect, but rather the total amount. Our model captures these reward sources; we define miner utilities explicitly as their expected reward per unit of time. Furthermore, our profitability analyses are more nuanced as they must directly consider the specific inter-block times, which requires explicitly modeling the orphan rate and its implied block production rate.

\paragraph*{Independence of randomness sources.}
There are three sources of randomness in the NCG ($\vec{t},\vec{m},r$), drawn independently prior to the game. It is not obvious that we can assume independence without loss of generality since the block production rate is a function of the orphan rate $\lambda$, which is determined by strategies of miners, which in turn depend on $r$. Crucially, the assumption that the orphan rate is stable for the entire duration of the game eliminates this dependence.
Note that such independence of miner strategies and time might not be present in other consensus games. In Proof-of-Stake, for example, the leader is elected for a fixed duration and may choose to delay their block publication intentionally to capture extra rewards – see \Cref{ex:timinggames} for a discussion of these ``timing games.''

\section{Omitted Proofs}\label{app:omitted}

\subsection{Proof of \Cref{lem:lambda}}\label{pr:lambda}
\paragraph*{\textbf{Lemma statement: }}
    Let $\lambda$ measure the probability that a block produced in the Markov Chain is orphaned. Then,
    \begin{align*}
        \lambda = p_1 (1-\alpha) \left(1+\frac{\alpha}{1-2\alpha}\right).
    \end{align*}
\begin{proof}
    Every time the Markov Chain enters \texttt{State 0'}, a block is orphaned. Additionally, for all \texttt{State i} where $i \geq 2$, a block is orphaned with probability $1-a$ as any honest block will be abandoned when \texttt{State 0''} is reached. Thus,
	\begin{align*}
		\lambda &= p_{0'} + (1-a)\sum_{i=2}^\infty p_i \\ 
        &= p_1 (1-\alpha)\left(1 + \sum_{i=2}^\infty \left(\frac{\alpha}{1-\alpha}\right)^{i-1}\right) \\
        &= p_1 (1-\alpha) \left(1+\frac{\alpha}{1-2\alpha}\right).
	\end{align*}
\end{proof}

\section{Extended derivations}\label{app:derivations}

\subsection{Stationary distribution calculation}\label{app:stationary}
Using this Markov Chain, we calculate the stationary distribution using the same technique conducted in Appendix~E.2 in \cite{carlsten2016instability}.
\begin{definition}[Stationary distribution, $p_i$]\label{def:stationary}
    Let $p_i$ denote the stationary distribution of the Markov Chain for \texttt{State i}.
    We start by calculating all probabilities relative to $p_1$,
    \begin{alignat*}{2}
    	&p_0 &&= \frac{p_1}{\alpha \int_0^\infty \frac{e^{-t/(1-\lambda)}}{(1-\lambda)}F_t(\beta)dt} \\ 
    	&p_{0'} &&= p_1(1-\alpha) \\
    	&p_{0''} &&= p_1 \alpha \\
    	&p_i &&= p_1\left(\frac{\alpha}{1-\alpha}\right)^{i-1}, \; \text{for } i \geq 1.
    \end{alignat*}
    Using the simplex constraint, we solve for $p_1$ explicitly,
    \begin{align*}
    	p_0 &+ p_{0'} + p_{0''} + \sum_{i=1}^\infty p_i =1 \\
        &\implies  p_1 = \left(\frac{1}{\alpha \int_0^\infty \frac{e^{-t/(1-\lambda)}}{(1-\lambda)}F_t(\beta)dt} + 1 + \frac{1-\alpha}{1-2\alpha}\right)^{-1}.
    \end{align*}
\end{definition}

\subsection{Deriving $f_i$ }\label{app:general-f0-f1}

Building on \Cref{ex:state3paths} and \Cref{def:attack-paths}, we generalize for \texttt{State i} where $i \geq 2$.

\begin{lemma}[$f_{i\geq 2}$]\label{thm:attackgeq2}
	For all states $i\geq 2$, the expected attacker rewards collected in \texttt{State i},
	\begin{align*}
		f_i &= \sum_{j=0}^{i-1} \left[\alpha (1-\alpha)^j \int_{0}^\infty \frac{t^j e^{-t/(1-\lambda)}}{j!(1-\lambda)^{j+1}}\mathbb{E}_{r}[R(t)]dt\right]
	\end{align*}
\end{lemma}

\begin{proof}
	In the set of \texttt{State i} attacker paths, there is exactly one path for each length $j=1, 2, \ldots i$, and the paths are $j-1$ copies of \texttt{H} before a single \texttt{A} (\texttt{A, HA, HHA, ...}). Each path occurs with probability $\alpha(1-\alpha)^j$, and the distribution of time for the length of the path is Erlang$(j,1/(1-\lambda))$. We integrate over the density of these path timings and multiply by the expectation of $R(t)$ over all remaining randomness, $r$. 
\end{proof}

\paragraph*{Calculating $f_0$.}
\texttt{State 0} requires deriving the expected reward for an attacker, given they may or may not hide a block they find.
From \texttt{State 0}, rewards are canonicalized by an attacker block in three ways:
\begin{description}[leftmargin=!,labelwidth=2.5cm]
	\item[Case i] the block has more rewards than $\beta$ (the attacker publishes),
	\item[Case ii] the block has less rewards than $\beta$ (the attacker hides) \textit{and} the attacker finds the next block,
	\item[Case iii] the block has less rewards than $\beta$ \textit{and} honest finds the next block (transitioning to \texttt{State 0'}) \textit{and} the attacker fork wins the race.
\end{description}
We treat each case individually. For \textbf{Case i}, the attacker publishes the block and thus realizes those rewards immediately on the canonical chain.
\begin{align*}
	f_{0,(i)} &= \underbrace{\alpha}_{\shortstack{\scriptsize attacker block}}\int_{t=0}^\infty 
	\underbrace{\frac{e^{-t/(1-\lambda)}}{(1-\lambda)}}_{\shortstack{\scriptsize density of time}}
	\underbrace{\int_{x=\beta}^\infty x F'_t(x)  dx}_{\shortstack{\scriptsize expected reward $\geq \beta$\\ \scriptsize at time $t$}} dt.
\end{align*}
This is exactly the expected attacker value of the state transition \texttt{State 0} $\rightarrow$ \texttt{State 0}. The inner integral bounds are $\beta \to \infty$ to capture the expected rewards given they are greater than $\beta$. For \textbf{Case ii}, the attacker block mined in \texttt{State 0} will become canonicalized for certain once they mine the second block. Thus, their rewards are realized when they transition to \texttt{State 2}.
\begin{align*}
	f_{0,(ii)} &= \underbrace{\alpha^2}_{\shortstack{\scriptsize two attacker \\ \scriptsize blocks}}
	\int_0^\infty 
	\underbrace{\frac{e^{-t/(1-\lambda)}}{(1-\lambda)}}_{\shortstack{\scriptsize density of time}}
	\underbrace{\int_{x=0}^\beta x F'_t(x)  dx}_{\shortstack{\scriptsize expected reward $< \beta$\\ \scriptsize at time $t$}} dt.
\end{align*}
This is the contribution to the attacker's expected rewards of the state transition \texttt{State 0} $\rightarrow$ \texttt{State 1} given a second attacker block in a row. Here, the integral is evaluated from $0 \to \beta$ to account for the expected value of rewards conditioned on the block remaining unpublished. For \textbf{Case iii}, the attacker block mined in \texttt{State 0} will become canonicalized if they win the race out of \texttt{State 0'} (e.g., either themselves or the $\gamma (1-\alpha)$ portion of the honest network that contributes to their chain mining the subsequent block and breaking the tie). Thus, their rewards are realized when they transition back to \texttt{State 0}.
\begin{align*}
f_{0,(iii)} 
&= \underbrace{\alpha}_{\shortstack{\scriptsize attacker block \\ \scriptsize in \texttt{State 0}}}
   \underbrace{(1-\alpha)}_{\shortstack{\scriptsize honest block \\ \scriptsize in \texttt{State 1}}}
   \underbrace{(\alpha + \gamma(1-\alpha))}_{\shortstack{\scriptsize attacker fork \\ \scriptsize wins tie-break}} \cdot \int_0^\infty 
   \underbrace{\frac{e^{-t/(1-\lambda)}}{(1-\lambda)}}_{\shortstack{\scriptsize density of time}}
   \underbrace{\int_{x=0}^\beta x F'_t(x)\,dx}_{\shortstack{\scriptsize expected reward $< \beta$\\ \scriptsize at time $t$}}\,dt.
\end{align*}
Thus $f_0 = f_{0,(i)}+f_{0,(ii)}+f_{0,(iii)}$.

\paragraph*{Calculating $f_1$.}
For \texttt{State 1}, rewards arriving in that state will be canonicalized by the attacker under two paths: (i) the attacker finds the next block (transitioning into \texttt{State 2}) or (ii) the honest party finds the next block (transitioning into \texttt{State 0'}) \textit{and} the attacker finds the subsequent. This is the same as the for the \texttt{State 2} attacker paths \texttt{A,HA}, so we use \Cref{thm:attackgeq2} with $i=2$,
\begin{align*}
	f_1 =& \alpha \int_0^{\infty} \frac{e^{-t/(1-\lambda)}}{(1-\lambda)} \mathbb{E}_{r}[R(t)]dt + \alpha(1-\alpha) \int_{0}^\infty \frac{te^{-t/(1-\lambda)}}{(1-\lambda)^2}  \mathbb{E}_{r}[R(t)] dt.
\end{align*}
Note that for \texttt{States 0', 0''}, the rewards accrued are already accounted for in $f_1$ and $f_2$ calculations, respectively. With $\lambda$ derived in \Cref{lem:lambda}, the stationary distribution calculated in \Cref{def:stationary} (the $p_i$ values), and the per-state attacker expected rewards calculated in \Cref{def:fis} (the $f_i$ values), we can calculate the full rewards for the attacker following the $\beta-$cutoff strategy.

\subsection{Transition probabilities}\label{app:transition-probs-instance}

We instantiate the general Markov Chain (\Cref{def:markovchain}) with our reward function $\hat{R}$. Recall that the selfish miner hides their block in \texttt{State 0} only if the realized rewards of the block are less than $\beta$. We calculate the CDF of the reward function (\Cref{def:cdf}), which depends on the relative size of $\beta$ and $E+C$.\footnote{We ignore the case where $C > \beta$ because that implies the attacker never hides their block and mines honestly.} If $\beta \leq C+E$, then the Bernoulli trial succeeding means $R(t) = t + C + E > \beta, \forall t$. Thus, for a given amount of time since parent, $t$, the total reward is less than $\beta$ only if the trial fails,
\begin{align*}
	F_t(\beta)_{\beta \leq E+C} = 
	\begin{cases}
		1-p & \text{if } t \leq \beta-C \\
		0 & \text{otherwise} 
	\end{cases}
\end{align*} 
If $\beta > E+C$, the total rewards may be less than $\beta$ even if the trial succeeds. Thus, the time component of the rewards must be sufficiently large for the total reward to exceed $\beta$. First, if $t < \beta-C-E$, the total rewards are certainly less than $\beta.$ If $t \in [\beta-C-E, \beta-C]$, the total reward is greater than $\beta$ only if the Bernoulli trial succeeds. Lastly, if $t \geq \beta - C$, the rewards exceed $\beta$ regardless of the trial outcome. Thus, 
\begin{align*}
	F_t(\beta)_{\beta > E+C} = 
	\begin{cases}
		1 & \text{if } t < \beta - C - E \\
		1-p & \text{if } t \in [\beta-C-E, \beta-C] \\
		0 & \text{otherwise}
	\end{cases}
\end{align*}

\begin{figure}
    \centering
    \includegraphics[width=0.85\linewidth]{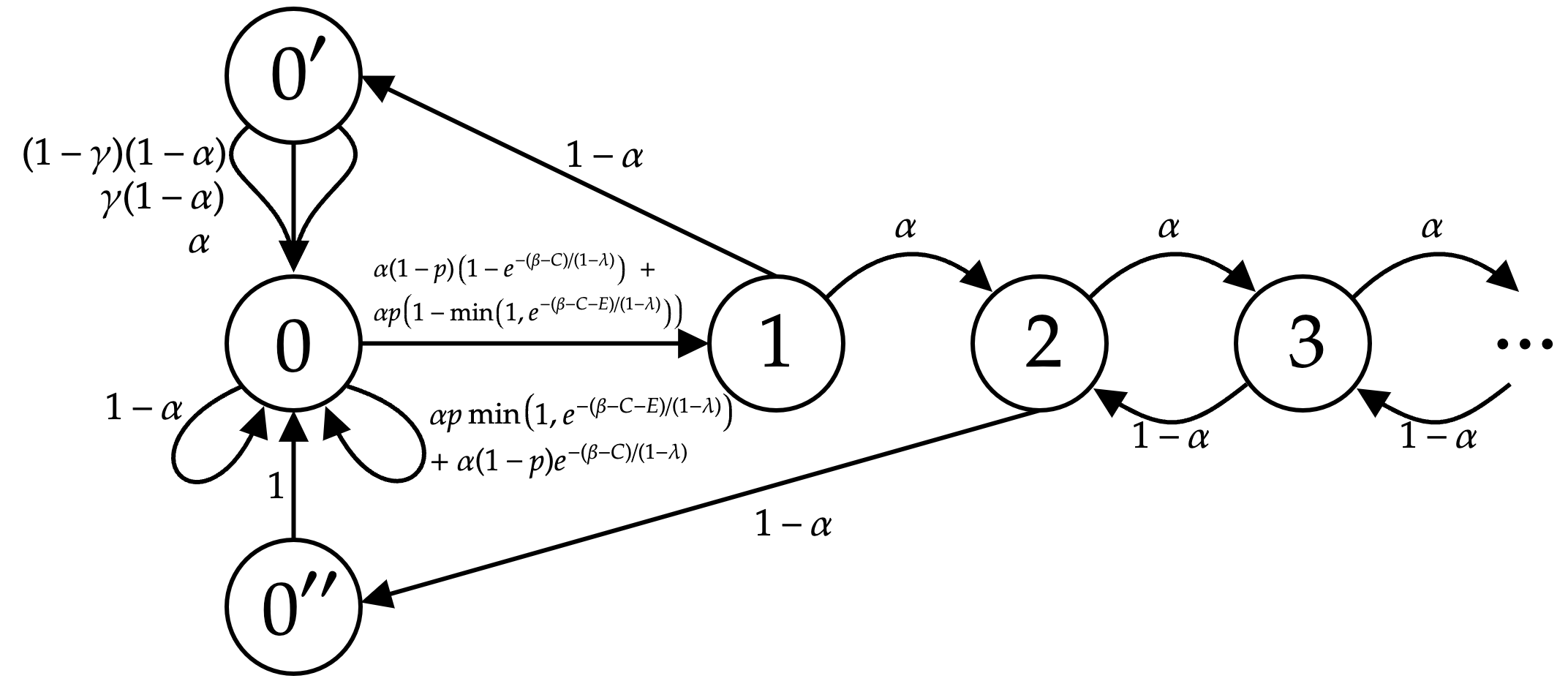}
    \caption{A Markov Chain for the $\beta-$cutoff strategy under the combination of (i) deterministic linear-in-time transaction fees, (ii) block rewards of magnitude $C$, and (iii) an extra Bernoulli reward of magnitude $E$. The $\min$ function is necessary to capture both the cases of $\beta \leq C + E$ and $\beta > C+E$ }
    \label{fig:mc-instance}
\end{figure}

Using these CDFs, we start by calculating the probability of the attacker transitioning to \texttt{State 1} (as in \Cref{eq:zerotoone}). 
An attacker will hide a block at time $t$ if the total rewards of the block are less than $\beta$. We calculate this probability by integrating over all possible times.
\begin{align*}
\Pr[\texttt{State 0} \rightarrow \texttt{State 1}] 
&= \alpha \Big[
   \underbrace{(1-p)\left(1-e^{-(\beta - C)/(1-\lambda)}\right)}_{\shortstack{\scriptsize trial fails and\\ \scriptsize $t < \beta - C$}} + \underbrace{p\left(1-\min\left(1,e^{-(\beta - C - E)/(1-\lambda)}\right)\right)}_{\shortstack{\scriptsize trial succeeds and\\ \scriptsize $t < \beta - C - E$}}
\Big].
\end{align*}

Next, we calculate the other transition out of \texttt{State 0}, where the attacker publishes their block because the reward exceeds $\beta$ (as in \Cref{eq:zerotozero}). An attacker will publish a block at time $t$ since the parent block if the total rewards of the block are greater than $\beta$. We calculate this probability by integrating over all possible times.
\begin{align*}
	\Pr[&\texttt{State 0} \rightarrow \texttt{State 0} \land \text{attacker block}] \\
	&= \alpha \Big[
	\underbrace{(1-p)e^{-(\beta -C)/(1-\lambda)}}_{\shortstack{\scriptsize trial fails and \\\scriptsize $t \geq \beta - C$}} + 
	\underbrace{p \min\left(1,e^{-(\beta-C-E)/(1-\lambda)}\right)}_{\shortstack{\scriptsize trial succeeds and \\\scriptsize $t \geq \beta - C - E$}}\Big].
\end{align*}

With these state transitions calculated, we present the complete Markov chain for the $\beta-$cutoff strategy in \Cref{fig:mc-instance}. 

\paragraph*{Stationary distribution}
Using \Cref{def:stationary}
\begin{align*}
&p_0 
= \frac{p_1}{
  \alpha (1-p)\bigl(1-e^{-(\beta-C)/(1-\lambda)}\bigr) 
  + \alpha p \bigl(1 - \min\bigl(1, e^{-(\beta-C-E)/(1-\lambda)}\bigr)\bigr)
} 
\\
&p_1 
= \biggl[
   \frac{1}{
     \alpha (1-p)\bigl(1-e^{-(\beta-C)/(1-\lambda)}\bigr)
     + \alpha p \bigl(1 - \min\bigl(1, e^{-(\beta-C-E)/(1-\lambda)}\bigr)\bigr)
   }
   + 1
   + \frac{1-\alpha}{1-2\alpha}
\biggr]^{-1}.
\end{align*}

\subsection{Expected attacker rewards}\label{app:expected-rewards-instance}
To continue the attacker reward calculation, we need to calculate the per-state expected attacker reward (\Cref{def:fis}). To calculate these values, we need to find the expected value of the reward function, $\hat{R}$ (\Cref{eq:fullrews}), depending on the time until the next block. Again, we use the \texttt{State 3} example to illustrate.
\begin{example}[\texttt{State 3} attacker paths, $\hat{R}$]
Recall that we have paths, \texttt{A, HA, HHA} respectively. Each block the attacker creates earns the constant block reward, $C$, and a Bernoulli reward of magnitude, $p \cdot E$. 
\begin{align*}
    f_3 =& \underbrace{(C + p\cdot E)\cdot(\alpha + (1-\alpha)\alpha + (1-\alpha)^2\alpha)}_{\shortstack{\scriptsize block and Bernoulli rewards}} + \underbrace{(1-\lambda)\cdot( \alpha + 2(1-\alpha)\alpha + 3 (1-\alpha)^2\alpha)}_{\shortstack{\scriptsize linear-in-time transaction fees}}.
\end{align*}
\end{example}
For the derivation according to \Cref{thm:attackgeq2}, see \Cref{app:f3derived}. This example prompts the instantiated versions of $f_{i}$. The expected attacker reward in \texttt{State i}, where $i \geq 2$, is
\begin{align*}
    f_{i\geq 2} = \underbrace{(C + p\cdot E)\cdot \sum_{j=0}^{i-1}\left[\alpha (1-\alpha)^i\right]}_{\text{block and bernoulli rewards}} 
    + \underbrace{(1-\lambda)\sum_{j=0}^{i-1}\left[ \alpha (1-\alpha)^{j} (j+1)\right]}_{\text{linear-in-time transaction fees}}
\end{align*}

This follows from enumerating the $i$ paths out of \texttt{State i} and calculating the probability of each occurring multiplied by the expected length of that path to find the value of the reward function. Note that we can write $\mathbb{E}_r[R(t)] = C + p\cdot E + t,$ because the expectation over the randomness of the Bernoulli reward is the expected value of the trial and the expectation over the time reward is linear as $t$.

\paragraph*{Calculating $f_0$.}
As in \Cref{subsec:perstateattacker}, we enumerate the three cases for \texttt{State 0}. We first define the PDF of $\hat{R}(t)$, 
\begin{align*}
	F'_t(x) = (1-p)\cdot\frac{e^{-(x-C)/(1-\lambda)}}{(1-\lambda)} + p\cdot \frac{e^{-(x-C-E)/(1-\lambda)}}{(1-\lambda)}.
\end{align*}
At time $t$, the instantaneous probability that the reward function $\hat{R}(t) = x$ depends on the outcome of the Bernoulli trial. If the trial fails, then the total reward is $\hat{R}= t+C$; thus $\Pr[t+C] = x$ is simply $\Pr[t] = x-C$, which for an exponential is $\frac{e^{-(x-C)/(1-\lambda)}}{(1-\lambda)}$. If the trial succeeds, by the same logic, we calculate $\Pr[t]=x-C-E$ as $\frac{e^{-(x-C-E)/(1-\lambda)}}{(1-\lambda)}.$ For \textbf{Case i}, the attacker publishes the block immediately; those rewards become theirs on the canonical chain.
\begin{align*}
	f_{0,(i)}
	&=\alpha \int_0^\infty\frac{e^{- t/(1-\lambda)}}{(1-\lambda)} \cdot\int_\beta^\infty x\left[(1-p)\cdot\frac{e^{-(x-C)/(1-\lambda)}}{(1-\lambda)} + p\cdot\frac{e^{-(x-C-E)/(1-\lambda)}}{(1-\lambda)}\right] dxdt.
\end{align*}
To evaluate the integral see \Cref{app:f0iderived}.
For \textbf{Case ii}, the attacker block mined in \texttt{State 0} will become canonicalized for certain once they mine the second block. Thus, they realize these rewards when transitioning to \texttt{State 2}.
\begin{align*}
	f_{0,(ii)} 
	&=\alpha^2 \int_0^\infty \frac{e^{-t/(1-\lambda)}}{(1-\lambda)} \cdot\int_0^\beta x\left[(1-p)\cdot\frac{e^{-(x-C)/(1-\lambda)} }{(1-\lambda)}+ p\cdot \frac{e^{-(x-C-E)/(1-\lambda)}}{(1-\lambda)} \right] dxdt.
\end{align*}
To evaluate the integral, see \Cref{app:f0iiderived}. For \textbf{Case iii}, the attacker block mined in \texttt{State 0} will become canonicalized only if they win the race out of \texttt{State 0'} (i.e., either by themselves or the $\gamma (1-\alpha)$ portion of the honest network that contributes to their chain mining the subsequent block and breaking the tie). Thus, they realize these rewards upon transitioning to \texttt{State 0}.
\begin{align*}
    f_{0,(iii)} 
	&=\alpha(1-\alpha)(\alpha+\gamma(1-\alpha)) \cdot \int_0^\infty \frac{e^{-t/(1-\lambda)}}{(1-\lambda)} \\
    &\cdot \int_0^\beta x\left[(1-p)\cdot\frac{e^{-(x-C)/(1-\lambda)} }{(1-\lambda)}+ p\cdot \frac{e^{-(x-C-E)/(1-\lambda)}}{(1-\lambda)} \right] dxdt.
\end{align*}
For the evaluation of the integral see \Cref{app:f0iiiderived}.
Thus $f_0 = f_{0,(i)} + f_{0,(ii)}+f_{0,(iii)}$.

\paragraph*{Calculating $f_1$.}
To conclude, we need $f_1$. Rewards arriving in that \stateone will be canonicalized by the attacker under two paths: (i) the attacker finds the next block (transitioning into \texttt{State 2}) or (ii) the honest party finds the next block (transitioning into \texttt{State 0'}) \textit{and} the attacker finds the subsequent. This is \Cref{thm:attackgeq2} with $i=2$, 
\begin{align*}
	f_1 =  (C+p\cdot E)\cdot (\alpha + \alpha(1-\alpha)) + (1-\lambda)\cdot (\alpha + 2\alpha(1-\alpha)).
\end{align*}
As before, the rewards accruing in \texttt{States 0' \& 0''} are already accounted for in the reward calculations from \texttt{States 1 \& 2} respectively. We can now explicitly calculate the attacker reward (\Cref{def:fullreward}).
The full attacker reward under $\hat{R}$ is, 
\begin{align*}
	\text{ATTACKER REWARD} =& p_0f_0 + p_1f_1 +p_1\cdot \bigg(
    \underbrace{(C+p\cdot E)\cdot\frac{2\alpha^2(1-\alpha)}{1-2\alpha}}_{\text{bernoulli and block rewards}}+\underbrace{(1-\lambda)\cdot\frac{\alpha^2 (3-2\alpha)}{1-2\alpha}}_{\text{linear-in-time transaction fees}}
    \bigg)
\end{align*}
For the derivation, see \Cref{app:fullcombined}.

\subsection{Deriving $f_3$ under the combined rewards, $\hat{R}$}\label{app:f3derived}
Implementing \Cref{thm:attackgeq2} with $\hat{R}$ (\Cref{eq:fullrews})
\begin{align*}
    f_3 &= \sum_{j=0}^{2} \left[\alpha (1-\alpha)^j \int_{0}^\infty \frac{t^j e^{-t/(1-\lambda)}}{j!(1-\lambda)^{j+1}}\mathbb{E}_{r}[R(t)]dt\right] \\
    &= \sum_{j=0}^{2} \left[ \alpha (1-\alpha)^j \int_{0}^\infty \frac{t^j e^{-t/(1-\lambda)}}{j!(1-\lambda)^{j+1}} (C + p\cdot E + t)dt \right]\\ 
    &= \sum_{j=0}^{2} \left[ \alpha (1-\alpha)^j (C + p\cdot E) \right] +  \sum_{j=0}^{2} \left[\alpha (1-\alpha)^j \int_{0}^\infty \frac{t^{j+1} e^{-t/(1-\lambda)}}{j!(1-\lambda)^{j+1}} dt\right] \\ 
    &= \sum_{j=0}^{2} \left[\alpha (1-\alpha)^j (C + p\cdot E)\right] +  (1-\lambda)\sum_{j=0}^{2} \left[\alpha (1-\alpha)^j \cdot (j+1)\right]. 
\end{align*}
\vfill


\subsection{Deriving $f_{0,(i)}$ under the combined rewards, $\hat{R}$}\label{app:f0iderived}
\begin{align*}
	f_{0,(i)} =& \alpha \int_0^\infty \frac{e^{-t/(1-\lambda)}}{(1-\lambda)} \int_\beta^\infty x F'_t(x) dx dt \\
	=&\alpha \int_0^\infty \frac{e^{- t/(1-\lambda)}}{(1-\lambda)} 
    \int_\beta^\infty x\bigg[(1-p)\cdot \frac{e^{-(x-C)/(1-\lambda)}}{(1-\lambda)} \\
    &+ p\cdot \frac{e^{-(x-C-E)/(1-\lambda)}}{(1-\lambda)} \bigg] dxdt\\ 
	=&\alpha\Big[\underbrace{C \cdot \left(p\min\left(1,e^{-(\beta-C-E)/(1-\lambda)}\right)+(1-p)e^{-(\beta-C)/(1-\lambda)}\right)}_{\text{block reward}}\\
	&+ \underbrace{E \cdot \left(p \min\left(1,e^{-(\beta-C-E)/(1-\lambda)}\right)\right)}_{\text{bernoulli reward}}\\
	&+ \underbrace{ p \left(1-\lambda+\max(0,\beta- C-E)\right)\min\left(1,e^{-(\beta-C-E)/(1-\lambda)}\right)}_{\shortstack{\scriptsize expected time $\geq \beta-C-E$ \\\scriptsize given trial succeeded}} \\
	& +
	\underbrace{(1-p)\left(1-\lambda+\beta-C\right)e^{-(\beta -C)/(1-\lambda)}}_{\shortstack{\scriptsize expected time $\geq \beta-C$ \\\scriptsize given trial failed}}\Big]
\end{align*}
\newpage

\subsection{Deriving $f_{0,(ii)}$ under the combined rewards, $\hat{R}$}\label{app:f0iiderived}
\begin{align*}
	f_{0,(ii)} =& \alpha^2 \int_0^\infty \frac{e^{-t/(1-\lambda)}}{(1-\lambda)} \int_0^\beta x F'_t(x) dx dt \\
	=&\alpha^2 \int_0^\infty \frac{e^{-t/(1-\lambda)}}{(1-\lambda)} 
    \int_0^\beta x\bigg[(1-p)\cdot\frac{e^{-(x-C)/(1-\lambda)}}{(1-\lambda)} \\
    &+ p\cdot \frac{e^{-(x-C-E)/(1-\lambda)}}{(1-\lambda)} \bigg] dxdt\\ 
	=&\alpha^2 \Big[ \underbrace{C \cdot \bigg( p\left(1-\min\left(1,e^{-(\beta-C-E)/(1-\lambda)}\right)\right)}_{\text{block reward}}\\
    &+\underbrace{(1-p)\left(1-e^{-(\beta-C)/(1-\lambda)}\right)\bigg)}_{\text{block reward}}\\
	&+ \underbrace{E \cdot \left(p \left(1-\min\left(1,e^{-(\beta-C-E)/(1-\lambda)}\right)\right)\right)}_{\text{bernoulli reward}}\\
	&+ p\big(1-\lambda-\left(1-\lambda+\max(0,\beta- C-E)\right)\\
    &\;\cdot \underbrace{\min\left(1,e^{-(\beta-C-E)/(1-\lambda)}\right)\big)}_{\shortstack{\scriptsize expected time $< \beta$ \\\scriptsize given trial succeeded}}  \\
	&+ \underbrace{(1-p)\left(1-\lambda-\left(1-\lambda+\beta-C\right)e^{-(\beta -C)/(1-\lambda)}\right)}_{\shortstack{\scriptsize expected time $< \beta$ \\\scriptsize given trial failed}}\Big)\Big]
\end{align*}
\vfill

\subsection{Deriving $f_{0,(iii)}$ under the combined rewards, $\hat{R}$}\label{app:f0iiiderived}
\begin{align*}
	f_{0,(iii)} =& (1-\alpha)(\alpha+\gamma(1-\alpha))\alpha \int_0^\infty \frac{e^{-t/(1-\lambda)}}{(1-\lambda)} \int_0^\beta x F'_t(x) dx dt \\
	=&(1-\alpha)(\alpha+\gamma(1-\alpha))\alpha \\
	&\cdot \int_0^\infty \frac{e^{-t/(1-\lambda)}}{(1-\lambda)} \int_0^\beta x\bigg[(1-p)\cdot \frac{e^{-(x-C)/(1-\lambda)}}{(1-\lambda)} \\ 
    &\,+ p\cdot \frac{e^{-(x-C-E)/(1-\lambda)}}{(1-\lambda)} \bigg] dxdt\\ 
	=&(1-\alpha)(\alpha+\gamma(1-\alpha)) 
	\\
	&\cdot \Big[ \underbrace{C \cdot \bigg( p\left(1-\min\left(1,e^{-(\beta-C-E)/(1-\lambda)}\right)\right)}_{\text{block reward}}\\
    &+\underbrace{(1-p)\left(1-e^{-(\beta-C)/(1-\lambda)}\right)\bigg)}_{\text{block reward}}\\
	&+ \underbrace{E \cdot \left(p \left(1-\min\left(1,e^{-(\beta-C-E)/(1-\lambda)}\right)\right)\right)}_{\text{bernoulli reward}}\\
	&+ p\bigg(1-\lambda-\left(1-\lambda+\max(0,\beta- C-E)\right) \\
    &\,\cdot\underbrace{\min\left(1,e^{-(\beta-C-E)/(1-\lambda)}\right)\bigg)}_{\shortstack{\scriptsize expected time $< \beta$ \\\scriptsize given trial succeeded}}  \\
	&+ \underbrace{(1-p)\left(1-\lambda-\left(1-\lambda+\beta-C\right)e^{-(\beta -C)/(1-\lambda)}\right)}_{\shortstack{\scriptsize expected time $< \beta$ \\\scriptsize given trial failed}}\Big)\Big]
\end{align*}

\subsection{Deriving full attacker reward under $\hat{R}$}\label{app:fullcombined}
Starting with calculating the $p_i$ (\Cref{def:stationary}) and $f_i$ (\Cref{def:fis}) values, we have
\begin{align*}
	p_{i-1} &= p_1 \left(\frac{\alpha}{1-\alpha}\right)^{i-2}, \; i\geq 2 \\ 
	f_i &= \underbrace{(C+p\cdot E) \cdot \alpha \sum_{j=0}^{i-1} (1-   \alpha)^{j}}_{\text{bernoulli and block rewards}} + 
               \underbrace{(1-\lambda) \cdot \alpha \sum_{j=0}^{i-1} (1-\alpha)^{j} \cdot (j+1)}_{\text{linear-in-time transaction fees}} \\ 
    &= \underbrace{(C+p\cdot E)\cdot \left(1-(1-\alpha)^i\right)}_{\text{bernoulli and block rewards}}
    +\underbrace{(1-\lambda)\cdot \frac{1-(1+i\alpha)(1-\alpha)^i}{\alpha}}_{\text{linear-in-time transaction fees}}
\end{align*}
From \Cref{def:fullreward}, we write
\begin{align*}
    \text{ATTACKER REWARD} &= p_0f_0 + p_1f_1 + \alpha\sum_{i=2}^\infty  p_{i-1}f_i \\
    &= p_0f_0 + p_1f_1 +p_1\cdot \bigg(
    \underbrace{(C+p\cdot E)\cdot\frac{2\alpha^2(1-\alpha)}{1-2\alpha}}_{\text{bernoulli and block rewards}}+\underbrace{(1-\lambda)\cdot\frac{\alpha^2 (3-2\alpha)}{1-2\alpha}}_{\text{linear-in-time transaction fees}}
    \bigg)
\end{align*}

\section{Extended figures}\label{app:extended-figures}

\subsection{Linear-in-time transaction fees, block rewards, and Bernoulli rewards.}
We now turn to numerical results based on the expected attacker reward for the full combined reward function $\hat{R}$, which includes a Bernoulli reward. We chose $p=0.25, E=4$ to ensure that the expected Bernoulli reward ($p\cdot E= 1)$ matches the block reward (which is set to 1) and the expected linear rewards (scaled by $1/(1-\lambda)$ because of difficulty adjustment). 

\Cref{fig:rew-comp} shows the full attacker reward (\Cref{def:fullreward}) under the $\hat{R}$ reward function (\Cref{eq:fullrews}) for various strategies. For each value of $\alpha$, the $\beta$ is selected to maximize the portion of rewards denoted in the parenthesis (for \texttt{Selfish}, $\beta \to \infty$ as always hiding maximizes the share of block rewards). We see that optimizing for the \texttt{Total} reward function (the sum of the three constituent parts) dominates the other strategies for all values of $\alpha$. The inset axes zoom in on the critical region to show the values of $\alpha$ at which each strategy outperforms \texttt{Honest}. Note that \texttt{Honest} is represented by $3\alpha$ because the expected value of the sum of the reward sources is $3$. 

While \Cref{fig:interpolation} shows the rewards of a strategy \textit{when only accounting for a specific subset of the rewards}, \Cref{fig:rew-comp} instead evaluates each \textit{against the total reward function $\hat{R}$.} Thus, we see that the strategy which optimizes for the full set of rewards dominates each other strategy. Interestingly, the optimal $\beta-$cutoff when considering \textit{only linear rewards} performs nearly as well as the optimal $\beta$ given the full reward function. This is an interesting result as it demonstrates the relative importance of the size of the component rewards. Specifically, the miner optimizing for linear rewards gets pretty close to the optimal $\beta$, while the miner optimizing for Bernoulli rewards is markedly worse off when evaluated against the full set of rewards. 
\begin{figure}[h]
    \centering
    \includegraphics[width=0.8\linewidth]{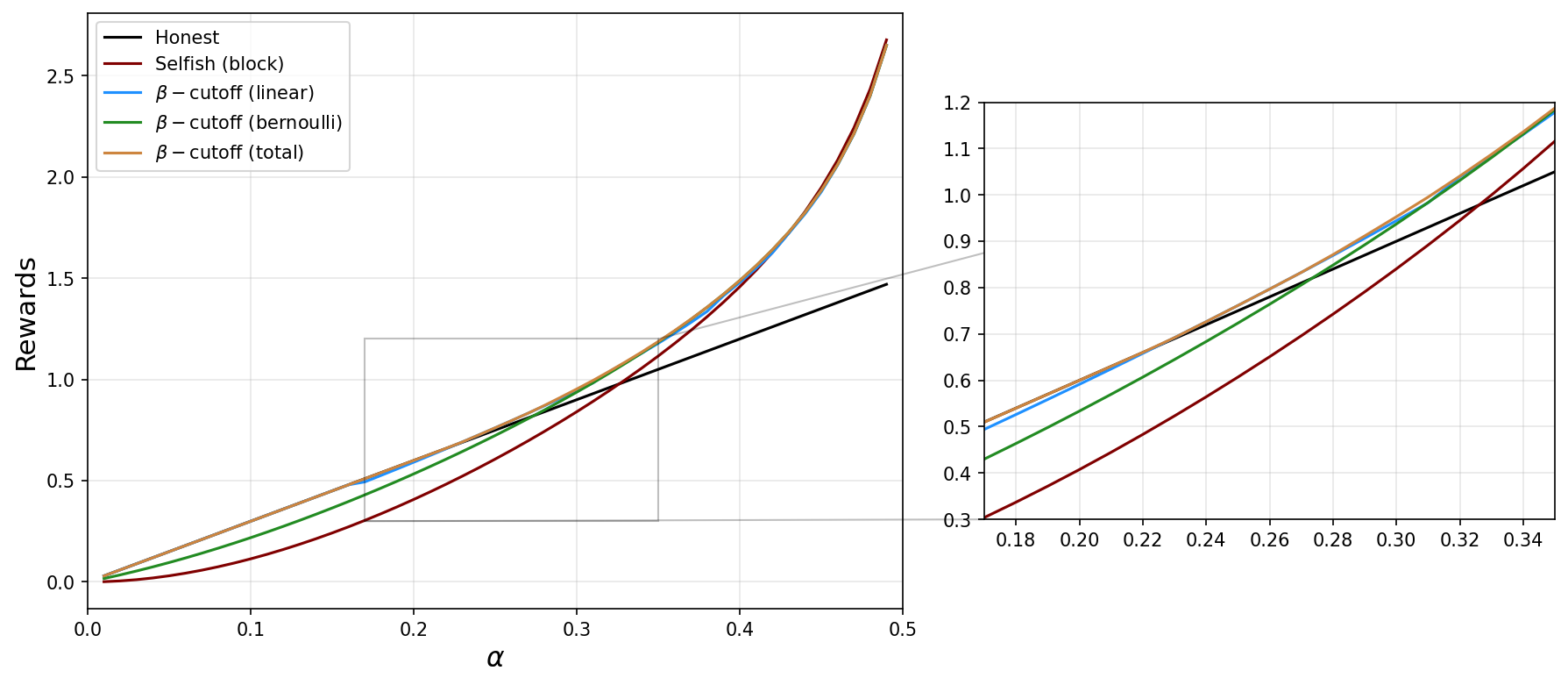}
    \caption{Comparing the full attacker reward (\Cref{def:fullreward}) under the $\hat{R}$ reward function (\Cref{eq:fullrews}) for various strategies with $p=0.25, E = 4, C=1, \gamma=0$. Each strategy chooses the $\beta$, which maximizes the reward portion described in parenthesis. We compare across a range of $\alpha$ values and see that optimizing for the total rewards dominates each of the other strategies, which focus on a single reward source.}
    \label{fig:rew-comp}
\end{figure}

\subsection{Rewards as a function of $\beta$ and simulation results.}

\Cref{fig:sims} plots the expected reward of each of the constituent rewards of $\hat{R}$ under various combinations of $\alpha,\beta$. Notably, the rewards may not be monotone in $\beta$, meaning the miner optimizing for the total rewards (or some specific subset) can choose the optimal $\beta$ that differs both from honest ($\beta=0$) and from selfish ($\beta \to \infty$). Each reward calculation for $\beta$-cutoff strategies in \Cref{fig:threshold-alphas,fig:interpolation,fig:bernoullis} chooses the optimal $\beta$ before evaluating the strategy against the benchmark. These simulated values help confirm that the path-counting technique presented in \Cref{sec:generalstatic} is correct. We also validate the methodology by performing a similar analysis for linear-in-time transaction fees and block rewards in \Cref{app:block-rews-only,app:lin-rews-only}, respectively.

\begin{figure}[h]
	\centering  
	\includegraphics[width=\linewidth]{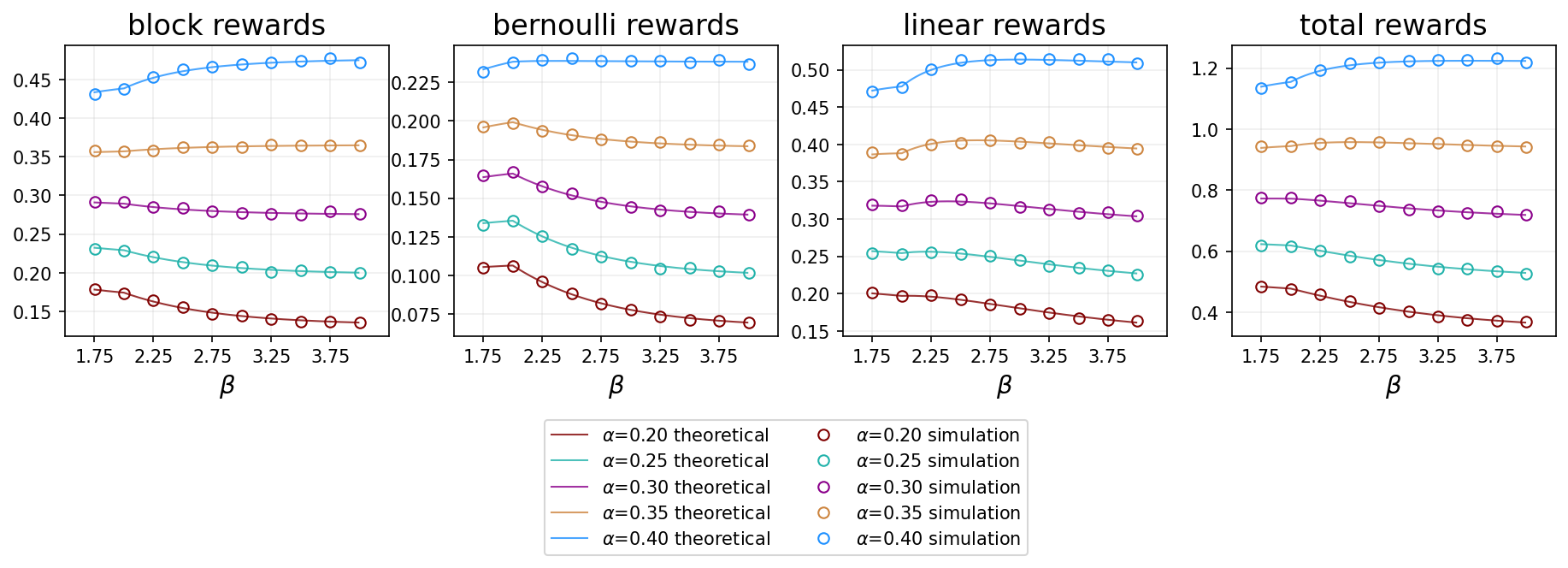}
	\caption{Theoretical (lines) and simulated values (dots) for miner rewards of the three component rewards constituting $\hat{R}$ as a function of $\alpha,\beta$. The fourth plot on the right is the sum of the others.}
	\label{fig:sims}
\end{figure}

\section{Worked example with only linear-in-time rewards}\label{app:lin-rews-only}
Consider only linear-in-time transaction fee rewards as in \cite{carlsten2016instability}, but with the $1/(1-\lambda)$ rate of block production. We confirm our results exactly analytically match the results of Appendix~E.2, despite using the path counting technique as opposed to their ``attacker probability of capturing each transaction'' method. With $R(t)=t$, the reward CDF (\Cref{def:cdf}) is simply,
\begin{align*}
    F_t(x) = 
    \begin{cases}
    1 & \text{if } t < x \\
    0 & \text{otherwise }
    \end{cases}
\end{align*}
Using the CDF, we derive the transition probabilities, which impact the stationary distribution (\Cref{def:stationary}).
\begin{align*}
    \Pr[\statezeronosp &\rightarrow \stateonenosp] \\
    &= \alpha \int_0^\infty 1/(1-\lambda)e^{-t/(1-\lambda)} F_t(\beta) dt\\   
    &= \alpha \int_0^\beta 1/(1-\lambda)e^{-t/(1-\lambda)} dt \\
    &= \alpha \left(1-e^{-\beta/(1-\lambda)}\right). \\ 
    \Pr[\statezeronosp &\rightarrow \statezeronosp \land \text{attacker block}] \\
    &= \alpha \int_0^\infty 1/(1-\lambda)e^{-t/(1-\lambda)} (1-F_t(\beta)) dt\\   
    &= \alpha \int_\beta^\infty 1/(1-\lambda)e^{-t/(1-\lambda)} dt \\
    &= \alpha \left(e^{-\beta/(1-\lambda)}\right).
\end{align*}
Now we need the the reward PDF (\Cref{def:cdf}),
\begin{align*}
    F'_t(x) = 1/(1-\lambda)e^{-x/(1-\lambda)}
\end{align*}
Using the PDF we calculate $f_0$ using the three cases.

\begin{figure}
    \centering
    \includegraphics[width=\linewidth]{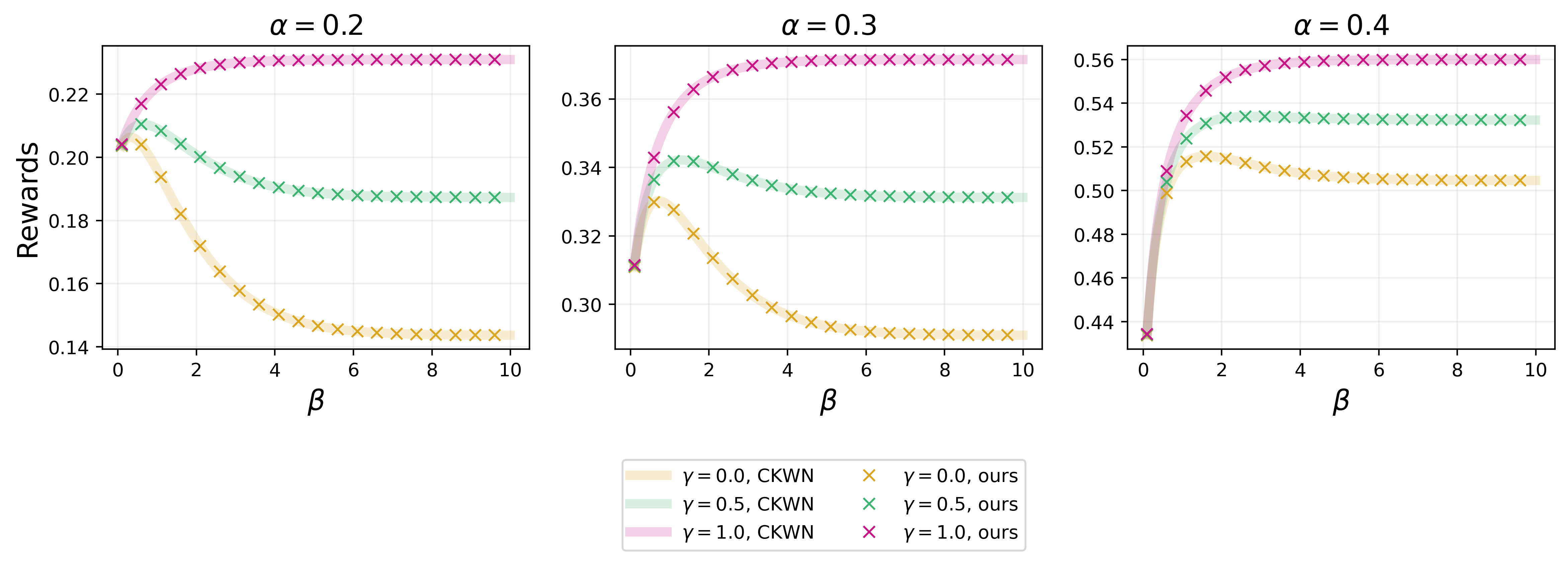}
    \caption{Comparing our analytic results (\Cref{eq:fullattacklinearonly} (denoted as \texttt{x}s labeled with \texttt{ours}) with the Appendix~E.2 formula from \cite{carlsten2016instability} (shown as lines labeled with \texttt{CKWN}). We show values for $\alpha=0.2,0.3,0.4$, $\lambda=0,1/2,1,$ and various values of $\beta$ – the values match to machine precision.} 
    \label{fig:linear-only}
\end{figure}

\noindent Case 1:
\begin{align*}
    f_{0,(i)} &= \alpha\int_\beta^\infty x/(1-\lambda)e^{-x/(1-\lambda)}dx \\
    &= \alpha e^{-\beta/(1-\lambda)}\left(\beta+1-\lambda\right).
\end{align*}
Case 2:
\begin{align*}
    f_{0,(ii)} &= \alpha^2 \int_0^\beta x/(1-\lambda)e^{-x/(1-\lambda)}dx \\ 
    &=\alpha^2 \left(1-\lambda - \left(\beta + 1-\lambda\right)e^{-\beta/(1-\lambda)}\right)
\end{align*}
Case 3:
\begin{align*}
    f_{0,(iii)} &= \alpha(1-\alpha)(\alpha+\gamma(1-\alpha)) \int_0^\beta x(1-\lambda)e^{-(1-\lambda)x}dx \\ 
    &=\alpha(1-\alpha)(\alpha+\gamma(1-\alpha)) \left(1-\lambda - \left(\beta + 1-\lambda\right)e^{-\beta/(1-\lambda)}\right)
\end{align*}
Given $k$ i.i.d. exponential random variables with rate $1/(1-\lambda)$, we have the sum of as $\text{Erlang}(k, 1/(1-\lambda))$, which has an expected value of $k(1-\lambda)$. Thus $\mathbb{E}_r[R(t)]$ for a length $k$ path is $k(1-\lambda)$. 
\paragraph*{Calculating $f_1$}
Using the definition of $f_1$,
\begin{align*}
    f_1 =& \alpha \int_0^{\infty} 1/(1-\lambda)e^{-t/(1-\lambda)} \mathbb{E}_{r}[R(t)]dt \\
    &+ \alpha(1-\alpha) \int_{0}^\infty 1/(1-\lambda)^2te^{-t/(1-\lambda)}  \mathbb{E}_{r}[R(t)] dt \\ 
    =& \left(1-\lambda\right)\cdot (\alpha + 2 \alpha(1-\alpha))
\end{align*}
Generalizing the above and following \Cref{thm:attackgeq2}, we have have
\begin{align*}
    f_{i\geq 2} &= \sum_{j=0}^{i-1} \left[\alpha (1-\alpha)^j \int_{0}^\infty \frac{t^j e^{-t/(1-\lambda)}}{(1-\lambda)^{j+1}j!}\mathbb{E}_{r}[R(t)]dt\right]\\
    &= \left(1-\lambda\right)\cdot \sum_{j=0}^{i-1} \alpha (1-\alpha)^j (j+1) \\
    &= \left(1-\lambda\right) \cdot \left(\frac{1-(i+1)(1-\alpha)^i+i(1-\alpha)^{i+1}}{\alpha}\right)
\end{align*}
Thus for the full attacker reward (\Cref{def:fullreward}), we have
\begin{align}\label{eq:fullattacklinearonly}
    &\text{ATTACKER REWARD} \\
    &\qquad= f_0 p_0 + f_1 p_1 \nonumber\\
    +& \alpha p_1 \sum_{i=2}^\infty \left(1-\lambda\right)\left[\left(\frac{1-(i+1)(1-\alpha)^i+i(1-\alpha)^{i+1}}{\alpha}\right) \cdot \left(\frac{\alpha}{1-\alpha}\right)^{i-2}\right] \nonumber\\ 
    &\qquad = f_0 p_0 + f_1 p_1 + p_1\left(1-\lambda\right) \cdot \frac{\alpha^2(3-2\alpha)}{1-2\alpha}
\end{align}
\Cref{fig:linear-only} shows the resulting rewards compared to the analytical result from Appendix~E.2 of \cite{carlsten2016instability}. These values match to machine precision.

\section{Worked example with only block rewards}\label{app:block-rews-only}
Consider the attacker maximizing only for their fraction of the block rewards as in \cite{eyal2013majority}. This ``purely selfish miner'' uses $\beta \to \infty$ as their $\beta$-cutoff strategy, such that they \textit{always} hide blocks mined in \statezeronosp.
With $R(t)=C$, the reward CDF (\Cref{def:cdf}) is simply,
\begin{align*}
    F_t(x) = 
    \begin{cases}
    1 & \text{if } C < x \\
    0 & \text{otherwise }
    \end{cases}
\end{align*}
Using the CDF, we derive the transition probabilities, which impact the stationary distribution (\Cref{def:stationary}) while taking the limit as $\beta \to \infty$, which simplifies the Markov Chain to Figure~1 in \cite{eyal2013majority},
\begin{align*}
    \Pr[\statezeronosp &\rightarrow \stateonenosp] \\
    &= \lim_{\beta\to\infty} \left(\alpha \int_0^\infty 1/(1-\lambda)e^{-t/(1-\lambda)} F_t(\beta) dt \right)\\   
    &= \alpha \\
    \Pr[\statezeronosp &\rightarrow \statezeronosp \land \text{attacker block}] \\
    &= \lim_{\beta\to\infty} \left(\alpha \int_0^\infty 1/(1-\lambda)e^{-t/(1-\lambda)} (1-F_t(\beta)) dt \right)\\   
    &= 0.
\end{align*}

\begin{figure}
    \centering
    \includegraphics[width=0.7\linewidth]{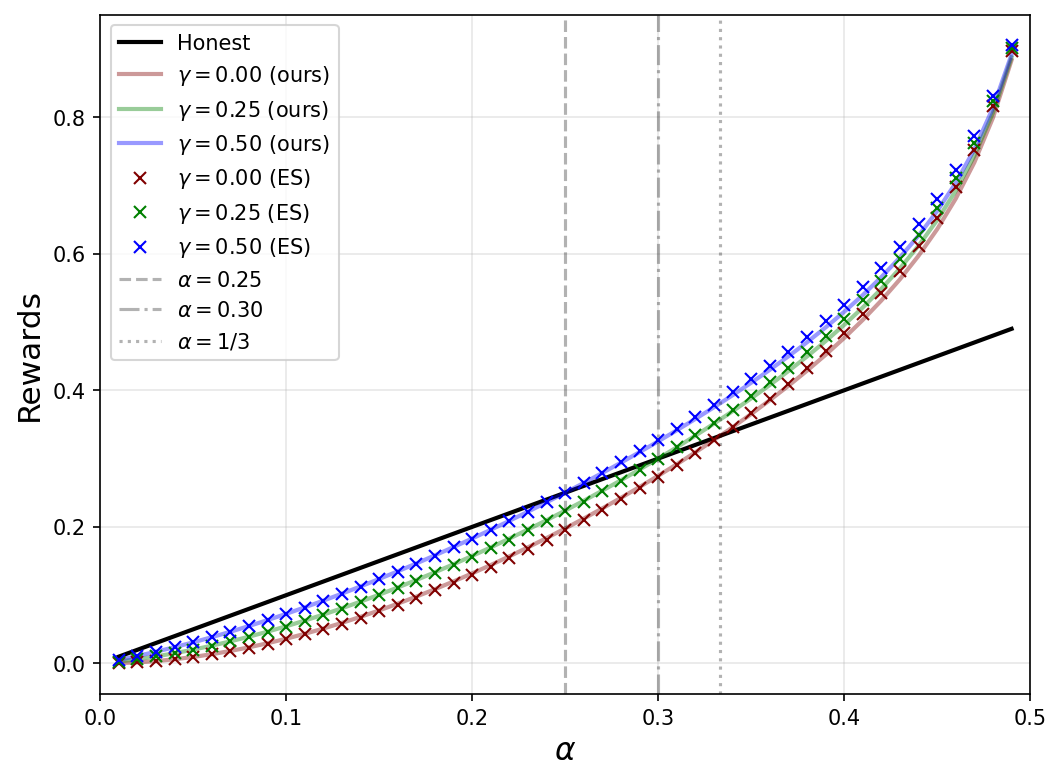}
    \caption{Comparing our analytic results (\Cref{eq:fullattackblockonly} (colored lines labeled with \texttt{ours}) with Equation~8 from \cite{eyal2013majority} (shown as \texttt{x}s labeled with \texttt{ES}). We show values for $\gamma=0,0.25,0.5$, various values of $\alpha$, and vertical lines at $0.25, 0.3,1/3$ (where $\gamma=0,0.25,0.5$ selfish mining respectively becomes profitable). The slight deviation at higher values of $\alpha$ arises from the introduction of \texttt{State 0''} (as in \cite{carlsten2016instability}).} 
    \label{fig:block-rew-only}
\end{figure}

\noindent We now derive the three cases for \statezeronosp. Case 1:
\begin{align*}
    f_{0,(i)} &= C\alpha.
\end{align*}
Case 2:
\begin{align*}
    f_{0,(ii)} &= C\alpha^2.
\end{align*}
Case 3:
\begin{align*}
    f_{0,(iii)} &= C \alpha(1-\alpha)(\alpha+\gamma(1-\alpha)).
\end{align*}
Since the block reward is constant at $C$, $\mathbb{E}_r[R(t)]=C$ for any length $k$ (recall that the attacker paths as defined in \Cref{ex:state3paths} each only have a \textit{single} attacker block).  
\paragraph*{Calculating $f_1$}
Using the definition of $f_1$,
\begin{align*}
    f_1 =& \alpha \int_0^{\infty} 1/(1-\lambda)e^{-t/(1-\lambda)} \mathbb{E}_{r}[R(t)]dt + \alpha(1-\alpha) \int_{0}^\infty 1/(1-\lambda)^2te^{-t/(1-\lambda)}  \mathbb{E}_{r}[R(t)] dt \\ 
    =& C \cdot (\alpha + \alpha(1-\alpha))
\end{align*}
Generalizing the above and following \Cref{thm:attackgeq2}, we have have
\begin{align*}
    f_{i\geq 2} &= \sum_{j=0}^{i-1} \left[\alpha (1-\alpha)^j \int_{0}^\infty \frac{t^j e^{-t/(1-\lambda)}}{(1-\lambda)^{j+1}j!}\mathbb{E}_{r}[R(t)]dt\right]\\
    &= C \cdot \sum_{j=0}^{i-1} \alpha (1-\alpha)^j \\
    &= C \cdot (1-(1-\alpha)^i).
\end{align*}
Thus for the full attacker reward (\Cref{def:fullreward}), we have
\begin{align}\label{eq:fullattackblockonly}
    \text{ATTACKER REWARD} =& f_0 p_0 + f_1 p_1 + \alpha p_1 \sum_{i=2}^\infty C\left[ (1-(1-\alpha)^i)\cdot \left(\frac{\alpha}{1-\alpha}\right)^{i-2}\right] \nonumber\\ 
    =& f_0 p_0 + f_1 p_1 + p_1 C \cdot \frac{2\alpha^2(1-\alpha)}{1-2\alpha}
\end{align}
\Cref{fig:block-rew-only} shows the resulting rewards compared to Equation~8 of \cite{eyal2013majority}. These values match nearly exactly. The slight deviation at higher values of $\alpha$ arises from the introduction of \texttt{State 0''} (as in \cite{carlsten2016instability}), which forces the attacker to mine honestly for a single block after publishing their private chain.

\end{document}